\documentclass[preprint]{iucr} 
\pdfoutput=1
\journalcode{J}
\usepackage{graphicx}
\usepackage{amsfonts,amsmath,amssymb,amsthm,color}

\newtheorem{theorem}{Theorem}
\newtheorem{remark}{Remark}
\newcommand{\R}{\mathbb{R}}
\newcommand{\V}[1]{\R^#1}
\newcommand{\Vp}[1]{\R^#1_{\ge 0}}  
\newcommand{\M}[2]{\R^{#1 \times #2}}  
\newcommand{\itv}[2]{\langle #1,#2 \rangle}  
\newcommand{\norm}[1]{\|#1\|}
\newcommand{\abs}[1]{|#1|}
\newcommand{\inv}[1]{#1^{-1}}

\newcommand{\red}[1]{{#1}}
\begin{document}

\title{Error evaluation of partial scattering functions obtained from contrast variation small-angle neutron scattering}   
\shorttitle{Short Title}          

\cauthor[a]{Koichi}{Mayumi}{kmayumi@issp.u-tokyo.ac.jp}{5-1-5 Kashiwanoha, Kashiwa-Shi, Chiba, 277–8581, Japan}
\author[a]{Tatsuro}{Oda}
\author[b]{Shinya}{Miyajima}
\author[c]{Ippei}{Obayashi}
\author[d]{Kazuaki}{Tanaka}
\aff[a]{The Institute for Solid State Physics, The University of Tokyo, 5-1-5 Kashiwanoha, Kashiwa-Shi, Chiba, 277–8581, Japan}
\aff[b]{Faculty of Science and Engineering, Iwate University}
\aff[c]{Center for Artificial Intelligence and Mathematical Data Science, Okayama University}
\aff[d]{Global Center for Science and Engineering, Waseda University}

\shortauthor{Author and Co-author Names}


\maketitle                        

\begin{abstract}
Contrast variation small-angle neutron scattering (CV-SANS) is a powerful tool to evaluate the structure of multi-component systems by decomposing scattering intensities  \( I\) measured with different scattering contrasts into partial scattering functions  \( S\) of self- and cross-correlations between components. The measured \( I\) contains a measurement error, \(\Delta I\), and \(\Delta I\) results in an uncertainty of partial scattering functions, \(\Delta S\). However, the error propagation from \(\Delta I\) to  \(\Delta S\) has not been quantitatively clarified. In this work, we have established deterministic and statistical approaches to determine  \(\Delta S\) from  \(\Delta I\). We have applied the two methods to (i) computational data of a core-shell sphere and experimental CV-SANS data of (ii) clay/polyethylene glycol (PEG) aqueous solutions and (iii) polyrotaxane solutions, and have successfully estimated the errors of  \(S\). The quantitative error estimation of \(S\) offers us a strategy to optimize the combination of scattering contrasts to minimize error propagation.
\end{abstract}

\section{Introduction}
Contrast variation small-angle neutron scattering (CV-SANS) has been utilized to study the nano-structure of multicomponent systems, such as organic/inorganic composite materials~\cite{NCgel, filled_rubber, CaCo3}, self-assembled systems of amphipathic molecules~\cite{blockcopolymer, selfassemble}, complexes of biomolecules ~\cite{protein, biological_membrane }, and supramolecular systems~\cite{CVSANS_PR, model_PR}. For isotropic materials, 2-D SANS data is converted to 1-D scattering function \( I(Q) \), where \( Q \) is the magnitude of the scattering vector. In the case of \( p \)-components systems, the scattering function is a sum of partial scattering functions, \( S_{ij}(Q) \)~\cite{CVSANS}:
\begin{align}\label{eq:I_Q_1}
    I(Q) = \sum_{i=1}^{p} \rho_{i}^{2} S_{ii}(Q) + 2 \sum_{i<j} \rho_{i} \rho_{j} S_{ij}(Q)
\end{align}
where \(\rho_i\) is the scattering length density of the \(i\)th component, \(S_{ii}(Q)\) is a self-term corresponding to the structure of the \(i\)th component, and \(S_{ij}(Q)\) is a cross-term originated from the correlation between the \(i\)th component and \(j\)th component. On the assumption of incompressibility, Eq.~\eqref{eq:I_Q_1} can be reduced to the following equation~\cite{CVSANS}:
\begin{align}\label{eq:I_Q_2}
    I(Q) = \sum_{i=1}^{p-1} (\rho_{i} - \rho_{p})^{2} S_{ii}(Q) + 2 \sum_{i<j} (\rho_{i} - \rho_{p})(\rho_{j} - \rho_{p}) S_{ij}(Q).
\end{align}
For 3-components systems, in which two solutes (\(i=1, 2\)) are dissolved in a solvent (\(i=3\)), the scattering function is given as below:

\begin{align}
\label{3component}
I(Q) &= (\rho_1 - \rho_3)^2 S_{11}(Q) + (\rho_2 - \rho_3)^2 S_{22}(Q) + 2(\rho_1 - \rho_3)(\rho_2 - \rho_3)S_{12}(Q) \nonumber\\
&= \Delta\rho_1^2 S_{11}(Q) + \Delta\rho_2^2 S_{22}(Q) + 2\Delta\rho_1\Delta\rho_2 S_{12}(Q)
\end{align}

Here, \(\Delta \rho_i\) is the scattering length density difference between the \(i\)th solute and solvent. Based on Eq.~\eqref{3component}, by measuring \(I(Q)\) of \(N\) samples (\(N \geq 3\)) with different scattering contrasts (\(\Delta \rho_1\) and \(\Delta \rho_2\)), it is possible to determine the three partial scattering functions, \(S_{11}(Q)\), \(S_{22}(Q)\), and \(S_{12}(Q)\):
\begin{align}
\label{SVSANSmodel}
\begin{bmatrix}
I_1 \\
\vdots \\
I_N
\end{bmatrix}
=
\begin{bmatrix}
\Delta_1\rho_1^2 & \Delta_1\rho_2^2 & 2\Delta_1\rho_1\Delta_1\rho_2 \\
\vdots & \vdots & \vdots \\
\Delta_N\rho_1^2 & \Delta_N\rho_2^2 & 2\Delta_N\rho_1\Delta_N\rho_2
\end{bmatrix}
\begin{bmatrix}
S_{11} \\
S_{22} \\
S_{12}
\end{bmatrix}
\end{align}
From the calculated partial scattering functions, we can analyze the structure of each solute and cross-correlation between the two solutes. 
Despite the usefulness of CV-SANS, its application has been limited due to the complexity and uncertainty of the calculation. The experimentally obtained \( I(Q) \) has a statistical error \( \Delta I \), and therefore we should consider how \( \Delta I \) propagates to the error of \( S(Q) \):
\begin{align}\label{eq:LSE}  
\begin{bmatrix}
I_1 + \Delta I_1 \\
\vdots \\
I_N + \Delta I_N
\end{bmatrix}
=
\begin{bmatrix}
\Delta_1\rho_1^2 & \Delta_1\rho_2^2 & 2\Delta_1\rho_1\Delta_1\rho_2 \\
\vdots & \vdots & \vdots \\
\Delta_N\rho_1^2 & \Delta_N\rho_2^2 & 2\Delta_N\rho_1\Delta_N\rho_2
\end{bmatrix}
\begin{bmatrix}
S_{11} + \Delta S_{11} \\
S_{22} + \Delta S_{22} \\
S_{12} + \Delta S_{12}
\end{bmatrix}
\end{align}
However, as far as we know, the relationship between \(\Delta I\) and \(\Delta S\) has not been clarified. The contribution of this study is an estimation of the transition from \(\Delta I\) to \(\Delta S\). To achieve this, we adopted two approaches: deterministic and statistical error estimations.

The objective of deterministic error estimation is to analytically estimate the upper bounds of $|\Delta S_{11}|$, $|\Delta S_{12}|$, and $|\Delta S_{22}|$ in Eq. \eqref{eq:LSE}. This essentially amounts to quantitatively clarifying the sensitivity of \(S\) to the variations in \(I\). It is essential to recognize this relationship, because it has a direct impact on the precision required for observing \(I\). 
Theoretically, solving a least square problem is equivalent to multiplying its right-hand side vector by the Moore-Penrose inverse (e.g., \cite{GV}), which is a kind of generalized inverse of its coefficient matrix. 
Therefore, the Moore-Penrose inverse plays a vital role in this analysis, enabling a detailed examination of the impact of each input variable. This sheds light on the complex pathways of error propagation within a system. 
Such an error estimation has already been done in the contexts of verified numerical computation (see \cite{M}, e.g.), and its effectiveness has been confirmed. 

The objective of statistical error estimation is to establish the probabilistic description of \(\Delta S\) from data under some statistical assumptions.
Statisticians have long considered the problem of error estimation, such as interval estimation~\cite{intro_stat} and Bayesian statistics~\cite{bayesian_stat}. Recently, this field has been referred to as uncertainty quantification~\cite{uqtextbook} and has been widely studied.
This study applies basic Bayesian inference with a noninformative prior distribution to estimate \(\Delta S\).
The estimation assumes that the error in \(I\), namely \(\Delta I\), follows a normal distribution.
The framework is also used to examine the robustness of the estimation.


These approaches effectively capture the inherent uncertainties in the observational errors. The error bounds are accurately derived from the mathematical structure of Eq. \eqref{SVSANSmodel}, thereby offering significant insights into the reliability and precision of the computational results for the partial scattering functions. In this study, we applied our error estimation methods to (i) computational data of a core-shell sphere and experimental CV-SANS data of (ii) clay/polyethylene glycol (PEG) aqueous solutions and (iii) polyrotaxane solutions~\cite{CVSANS_PR}. Polyrotaxane is a topological supramolecular assembly, in which ring molecules are threaded onto a linear polymer chain. This study demonstrates the effectiveness of our method in quantifying uncertainties arising from the randomness of observational errors. 
In both error estimation approaches, the condition number of the coefficient matrix is a useful tool to evaluate the degree of error propagation (see Sections~\ref{sec:DEE} and \ref{sec:statistical} for detail).


\section{Methods}\label{sec:methods}

\subsection{Deterministic error estimation}\label{sec:DEE}
Despite the usefulness of CV-SANS, its application has been limited due to the statistical error $\Delta I$ associated with the experimentally obtained $I(Q)$. It is not well-studied how $\Delta I$ propagates to the error of $S(Q)$. To address this issue, this subsection presents a theory that clarifies how $\Delta I$ propagates to the error of $S(Q)$, $\Delta S$,  in a deterministic sense.  
The theory in this subsection is based on error analyses in numerical computations \cite{H} .

Define 
\begin{equation}\label{eq:notation-IAS}
\begin{aligned}
I &:=
\begin{bmatrix}
I_1 \\
\vdots \\
I_N 
\end{bmatrix}, \quad 
\Delta I := 
\begin{bmatrix}
\Delta I_1 \\
\vdots \\
\Delta I_N
\end{bmatrix}, \quad 
A := 
\begin{bmatrix}
\Delta_1\rho_1^2 & \Delta_1\rho_2^2 & 2\Delta_1\rho_1\Delta_1\rho_2 \\
\vdots & \vdots & \vdots \\
\Delta_N\rho_1^2 & \Delta_N\rho_2^2 & 2\Delta_N\rho_1\Delta_N\rho_2
\end{bmatrix}, \\ 
S &:= 
\begin{bmatrix}
S_{11} \\
S_{22} \\
S_{12} 
\end{bmatrix}, \quad 
\Delta S := 
\begin{bmatrix}
\Delta S_{11} \\
\Delta S_{22} \\
\Delta S_{12}
\end{bmatrix}. 
\end{aligned}
\end{equation}
Then, (\ref{eq:LSE}) can be written as 
\begin{equation}
I + \Delta I = A(S + \Delta S).     
\end{equation}
It is obvious that $S$ and $\Delta S$ are 3-dimensional vectors, and $A$ is an $N \times 3$ matrix. 
Sections~\ref{sec:DEE} and \ref{sec:statistical} generalize Eq. (\ref{eq:LSE}), and treat the case where $S$ and $\Delta S$ are $m$-dimensional vectors, and $A$ is a $N \times m$ matrix. 
To this end, we introduce notations used in Sections~\ref{sec:DEE} and \ref{sec:statistical}. 
For $v \in \V{n}$, let $v_i$ for $i = 1,\dots,n$ be the $i$-th component of $v$, and $\norm{v} := \sqrt{v_1^2 + \cdots +v_n^2}$. 
For $v,w \in \V{n}$, the inequality $v \le w$ means that $v_i \le w_i$ holds for all $i$. 
Let $\Vp{n} := \{v \in \V{n}\  |\  v \ge 0\}$. 
For $B \in \M{m}{n}$, let $b_{ij}$ and $\norm{B}$ be the $(i,j)$ element and 2-norm (see \cite{GV}, e.g.) of $B$, respectively. 
Suppose $v \le w$, and define 
$$
[v,w] := \left(\begin{array}{c}
[v_1,w_1] \\
\vdots \\
{[}v_n,w_n{]} \\
\end{array}\right), \quad 
\abs{v} := \left(\begin{array}{c}
\abs{v_1} \\
\vdots \\
\abs{v_n}
\end{array}\right), \quad \abs{B} := \left(\begin{array}{ccc}
\abs{b_{11}} & \cdots & \abs{b_{1n}} \\
\vdots & \ddots & \vdots \\
\abs{b_{n1}} & \cdots & \abs{b_{nn}} \\
\end{array}\right). 
$$
Denote the Moore-Penrose inverse of $B$ by $B^+$. 
When $m \ge n$ and $B$ has full column rank in particular, we have $B^+ = \inv{(B^TB)}B^T$, where $B^T$ denotes the transpose of $B$. 

We present Theorem~\ref{th:AE} for clarifying how $\Delta I$ propagates to $\Delta S$ in a deterministic sense. 
See Appendix \ref{sec:proofs} for its proof.
Theorem~\ref{th:AE} says that we can analytically estimate an upper bound on $\abs{\Delta S}$.  
\begin{theorem}\label{th:AE}  
Let $A \in \M{N}{m}$, $I,\Delta I \in \V{N}$, $\Delta J \in \Vp{N}$, $S,\Delta S \in \V{m}$, and $\Delta T := \abs{A^+}\Delta J$. 
Suppose that $N \ge m$, $AS = I$, $I + \Delta I = A(S + \Delta S)$, $\abs{\Delta I} \le \Delta J$, 
and $A$ has full column rank. 
It then follows that 
\begin{equation}\label{eq:AE}
\abs{\Delta S} \le \Delta T. 
\end{equation}
\end{theorem}
\begin{remark}\label{rm:AE}
In practice, we regard the standard deviation of the experimentally obtained $I(Q)$ as $\Delta J$
\end{remark}

Define the condition number ${\rm cond}(A)$ by ${\rm cond}(A) := \norm{A}\norm{A^+}$. 
Let $\lambda_{\max}$ and $\lambda_{\min}$ be the largest and smallest singular values of $A$, respectively. 
We then have $\norm{A} = \lambda_{\max}$, $\norm{A^+} = 1/\lambda_{\min}$, so that ${\rm cond}(A) = \lambda_{\max}/\lambda_{\min}$. 
It can be shown by using the singular value decomposition  (e.g., \cite{GV}) of $A^+$ that $\norm{A^+I} \ge \norm{I}/\lambda_{\max}$. 
These relations\red{,} $S = A^+I$ \red{and $\Delta S = A^+\Delta I$} give 

\begin{equation}\label{eq:error_condA}
\frac{\norm{\red{\Delta S}}}{\norm{S}} = \frac{\norm{\red{A^+\Delta I}}}{\norm{A^+I}} \le \frac{\norm{A^+}\norm{\red{\Delta I}}}{\norm{I}/\lambda_{\max}} =
\frac{\lambda_{\max}\norm{\red{\Delta I}}}{\lambda_{\min}\norm{I}} = \frac{{\rm cond}(A)\norm{\red{\Delta I}}}{\norm{I}} 
\end{equation}
if $S \ne 0$ and $I \ne 0$. 
This inequality implies that $\norm{\red{\Delta S}}/\norm{S}$ is enlarged by ${\rm cond}(A)$, suggesting that ${\rm cond}(A)$ is an important parameter related to the degree of the error propagation from $I$ to $S$. More detailed explanation of ${\rm cond}(A)$ is shown in Appendix \ref{sec:condnum}.

\subsection{Statistical error estimation}\label{sec:statistical}
To statistically quantify the estimation errors $\Delta S_{ij}$, we rewrite \eqref{eq:LSE} into the following statistical model:
\begin{equation}\label{eq:stat-model-1}
    I + \Delta I = A S,
\end{equation}
where
\begin{equation*}
\begin{aligned}
I &:=
\begin{bmatrix}
I_1 \\
\vdots \\
I_N 
\end{bmatrix}, \ 
\Delta I := 
\begin{bmatrix}
\Delta I_1 \\
\vdots \\
\Delta I_N
\end{bmatrix}, \  
A := 
\begin{bmatrix}
\Delta_1\rho_1^2 & \Delta_1\rho_2^2 & 2\Delta_1\rho_1\Delta_1\rho_2 \\
\vdots & \vdots & \vdots \\
\Delta_N\rho_1^2 & \Delta_N\rho_2^2 & 2\Delta_N\rho_1\Delta_N\rho_2
\end{bmatrix},  \ 
S := 
\begin{bmatrix}
S_{11} \\
S_{22} \\
S_{12} 
\end{bmatrix}. 
\end{aligned}
\end{equation*}
In this formulation, $S_{ij}$ includes $\Delta S_{ij}$ since the values we want to estimate are considered random variables with a prior distribution in the Bayesian framework.
The following assumptions are made to build the statistical model.
\begin{itemize}
    \item Each $\Delta I_i$ is a normal random variable with mean zero and standard deviation $\sigma_i$
    \item $\Delta I_1, \ldots, \Delta I_N$ are probabilistically independent
    \item The prior distribution of $S$ is a multivariate normal distribution $\mathcal{N}(0, \alpha^2 E)$, where $\alpha > 0$ is a parameter and $E$ is an $N\times N$ identity matrix
\end{itemize}
From the first two assumptions, $\Delta I$ is a multivariate normal random variable with mean zero and covariance matrix $\Sigma$, where
$$
\Sigma = \begin{bmatrix}
    \sigma_1^2 &  & \\
    & \ddots & \\
    & & \sigma_N^2
\end{bmatrix}.
$$
Then the posterior distribution of $S$ is $\mathcal{N}(\bar{S}_{\alpha}, \bar{\Sigma}_\alpha)$ from Bayes' formula for multivariate normal distributions [Section 6.1 in \cite{uqtextbook}], where $\bar{\Sigma}_\alpha = (\alpha^{-2} E + A^T \Sigma^{-1} A)^{-1}$ and $\bar{S}_\alpha = \bar{\Sigma}_\alpha A^T \Sigma^{-1} I$.
The prior distribution represents the assumption on the scale of $S$, and we try to remove the effect of the assumption using noninformative prior by $\alpha \to +\infty$.
As a result, the posterior distribution of $S$ is $\mathcal{N}(\bar{S}, \bar{\Sigma})$, where 
\begin{equation}\label{eq:posterior-1}
\begin{aligned}
    \bar{\Sigma}_\alpha \to \bar{\Sigma} &= (A^T \Sigma^{-1} A)^{-1}, \\
    \bar{S}_\alpha \to \bar{S} &= \bar{\Sigma} A^T \Sigma^{-1}I. 
\end{aligned}  
\end{equation}
By setting 
\begin{equation}
    \bar{S}  = \begin{bmatrix}
        \bar{S}_{11} \\ \bar{S}_{22}\\ \bar{S}_{12}
    \end{bmatrix}, \text{ and }
    \bar{\Sigma} = \begin{bmatrix}
        \bar{\sigma}_{11, 11} & \bar{\sigma}_{11, 22} & \bar{\sigma}_{11, 12} \\
        \bar{\sigma}_{22, 11} & \bar{\sigma}_{22, 22} & \bar{\sigma}_{22, 12} \\
        \bar{\sigma}_{12, 11} & \bar{\sigma}_{12, 22} & \bar{\sigma}_{12, 12}
    \end{bmatrix},
\end{equation}
the result can be interpreted as follows:
\begin{itemize}
    \item The posterior distribution of $S_{11}$ is a normal distribution with mean $\bar{S}_{11}$ and variance $\bar{\sigma}_{11, 11}$. 
    This means that $\bar{S}_{11}$ is the most likely value of $S_{11}$, but the uncertainty of the estimation is described by the normal distribution whose variance is $\bar{\sigma}_{11, 11}$
    \item $S_{22}$ and $S_{12}$ can be also evaluated in the same way 
    \item A non-diagonal element of $\bar{\Sigma}$ is the covariance of estimated values; that is, the estimated values are correlated
\end{itemize}

We remark that $\bar{S}$ differs from the solution of the standard least square problem, $A^+I$. 
In fact, $ \bar{S}$ is the solution of the weighted least square problem: $\mathrm{argmin}_S\left\|\Sigma^{-1}(I - AS)\right\|$. 
This formula means that $\sigma_j$ quantifies the reliability of measurements, and squared errors are weighted by the reliability factors.

The following theorem is useful for estimating the error before an experiment since singular values can be calculated only from $A$. See Appendix~\ref{sec:proofs} for its proof.
\begin{theorem}\label{th:Vij}
    $\sqrt{|\text{an element of } \bar{\Sigma}|} \leq (1/\lambda_{\min})\max_j \Delta I_j$, where $\lambda_{\min}$ is the smallest singular value of the matrix $A$.
\end{theorem}

Because the diagonal elements of $\bar{\Sigma}$ are the standard deviations of $S_{11}, S_{22}$, and $S_{12}$, we can say that the absolute error is roughly scaled by $1/\lambda_{\min}$. 
Because the observation $I$ is roughly equal to $A \bar{S}$, we can estimate the relation between $\|I\|$ and $\|\bar{S}\|$ as follows: 
\begin{equation}
    \|I\| \approx \|A\bar{S}\| \leq \|A\| \|\bar{S}\| = \lambda_{\max} \|\bar{S}\|,
\end{equation}
where $\lambda_{\max}$ is the largest singular value of the matrix $A$.
This means that $\|\bar{S}\|$ is approximately bounded by $ \|I\|/\lambda_{\max} $ from below, and we can say that the relative error is roughly scaled by $\lambda_{\max}/\lambda_{\min} = \mathrm{cond}(A)$. 
This suggests that condition numbers are useful for estimating the relative errors from the viewpoint of Bayesian statistics.

Using the formulation, we can examine the robustness of the estimation.
The assumptions on model~\eqref{eq:stat-model-1} are not perfect, and real measurement has unknown error factors, such as the uncertainty of the scattering length density, deviation of the noise distribution from the normal distribution, and unknown bias of the measurement device.
Of course, such errors are expected to be very small, but if these small errors significantly disturb the result, the estimated result will not be reliable.

One of the simplest ways to check the robustness of the result is to extend the error bars of the measurement virtually.
Here, we consider what happened when $\sigma_1, \ldots, \sigma_N$ are multiplied by $\mu > 1$.
In this case, $\Sigma$ is multiplied by $\mu^2$ in~\eqref{eq:posterior-1}, and as a result $\bar{\Sigma}$ is multiplied by $\mu^2$ but $\bar{S}$ is not changed.
Therefore, the standard deviations of the posterior distributions are multiplied by $\mu$, which means that the estimation's uncertainty is enlarged by $\mu$.

\subsection{Computational data of core-shell sphere}\label{subsec:method-cs}

To verify the validity of the deterministic and statistical error estimations, we first applied the two methods to computational data of a core-shell sphere (Fig.~\ref{fig:Fig1_cs} (a)). We computed scattering intensities, $I(Q)$, of core-shell spheres dispersed in D\(_{2}\)O/H\(_{2}\)O mixtures with different D\(_{2}\)O fractions by using the "Core shell sphere" model of the SAS view software, version 5.0.6 (http://www.sasview.org/). The core radius and shell thickness were 50 {\AA} and 10 {\AA}, respectively. While the scattering length densities of the core and shell were fixed at 4.0 and 1.0 \(\times 10^{-6}\) {\AA}\(^{-2}\), the scattering length density of the solvent was changed with D\(_{2}\)O fraction, \(\phi\rm{_{D}}\)b, as follows ~\cite{NCgel}:
\begin{equation}\label{eq:sld_solvent}
    \rho\rm{_{water}} = 6.95 \times \phi\rm{_{D}} - 0.56\   [\times 10^{-6} {\AA}^{-2}]
\end{equation}
The core-shell samples with \(\phi\rm{_{D}}\) = 1.0, 0.90, 0.80, 0.66, 0.50, 0.40, 0.22, 0.10, and 0.0 are named as CS100, CS090, CS080, CS066, CS040, CS022, CS010, and CS000, respectively (Fig.~\ref{fig:Fig1_cs} (b)).

\begin{figure}
    \centering
    \includegraphics[width=1.0\linewidth]{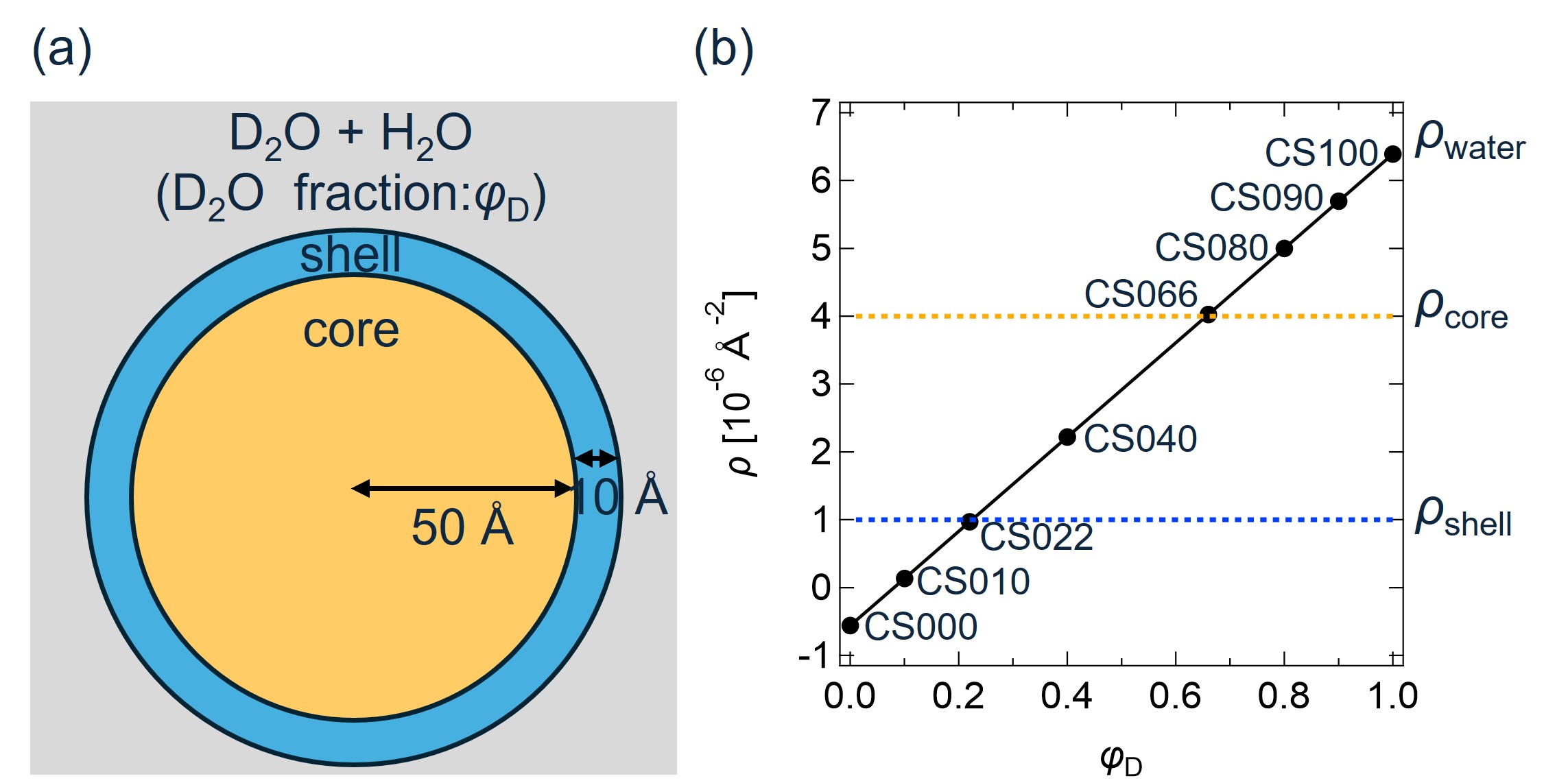}
    \caption{(a) Schematic illustration of a core-shell sphere dispersed in a D\(_{2}\)O/H\(_{2}\)O mixture; (b) scattering length densities of the core (\(\rho \rm{_{core}}\)), shell(\(\rho\rm{_{shell}}\)), and solvent(\(\rho\rm{_{water}}\)) plotted against D\(_{2}\)O fraction of the solvent, \(\phi\rm{_{D}}\).}.
    \label{fig:Fig1_cs}
\end{figure}

\subsection{CV-SANS data of clay/PEG aqueous solutions and polyrotaxane solutions}\label{subsec:method-cp-PR}
For the next step, the error estimation methods were applied to two sets of experimental CV-SANS data: (i) clay/polyethylenglycol (PEG) solutions (Fig.~\ref{fig:Fig2_cp} (a)) and (ii) polyrotaxane(PR) solutions (Fig.~\ref{fig:Fig3_PR} (a)). The clay/PEG solutions were prepared by dissolving Laponite XLG nanoclay (BYK) and PEG (\(M_{w}\) = 35,000, Fluka) in D\(_{2}\)O/H\(_{2}\)O mixtures with different D\(_{2}\)O fractions. 
According to a previous CV-SANS study on clay/PEG aqueous solutions~\cite{clayPEG}, PEG is adsorbed onto the surface of the clay particles, and a core-shell structure is formed as shown in Fig.~\ref{fig:Fig2_cp} (a). 
The volume fractions of clay and PEG were 2\(~\%\)  and 2.5\(~\%\), respectively. 
 The volume fraction of D\(_{2}\)O in the solvent, \(\phi_{D}\), was changed to vary the scattering contrasts of the clay and PEG.  Corresponding to \(\phi_{D}\) = 1.0, 0.80, 0.62, 0.40, 0.16, and 0.00, the clay/PEG solutions are named as CP100, CP080, CP062, CP040, Cp016, and CP000, respectively (Fig.~\ref{fig:Fig2_cp} (b)). The scattering length density of PEG, \(\rho\rm{_{PEG}}\), is \(0.65 \times 10^{-6} \rm{\AA}^{-2}\), while the scattering length densities of the D\(_{2}\)O/H\(_{2}\)O mixtures vary with  \(\phi_{D}\) as described in Eq.\eqref{eq:sld_solvent}. The scattering length density of the clay particle, \(\rho\rm{_{clay}}\), is given as follows: ~\cite{NCgel}: 
\begin{equation}\label{eq:sld_clay}
    \rho\rm{_{clay}} = 0.85 \times \phi\rm{_{D}} + 3.92\   [\times 10^{-6} {\AA}^{-2}]
\end{equation}

\begin{figure}
    \centering
    \includegraphics[width=1.0\linewidth]{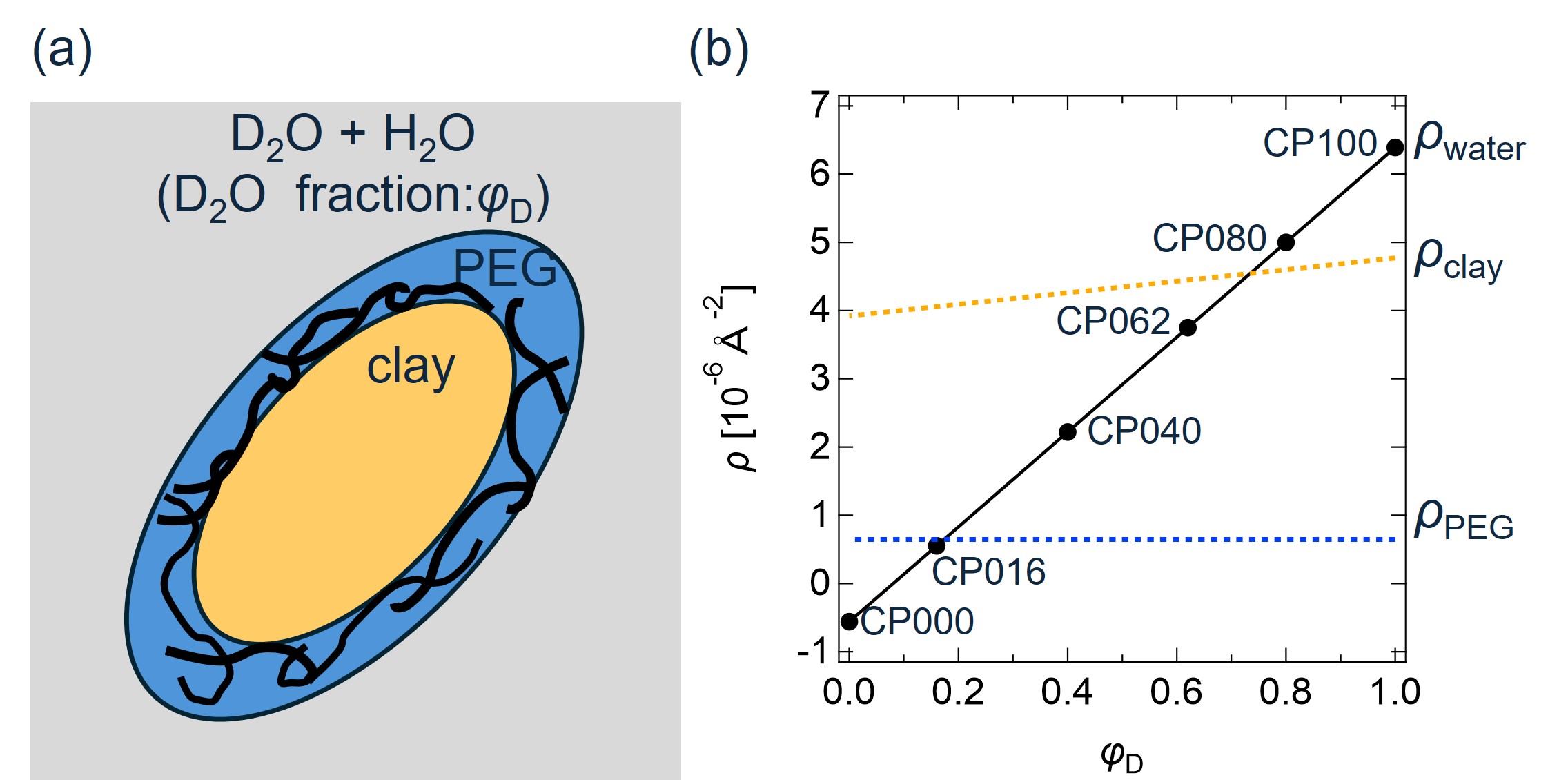}
    \caption{(a) Schematic illustration of a clay/PEG solution dissolved in a D\(_{2}\)O/H\(_{2}\)O mixture; (b) scattering length densities of the clay (\(\rho \rm{_{clay}}\)), PEG(\(\rho \rm{_{PEG}}\)), and solvent(\(\rho \rm{_{water}}\)) plotted against D\(_{2}\)O fraction of the solvent, \(\phi\rm{_{D}}\). }
    \label{fig:Fig2_cp}
\end{figure}

The CV-SANS data of PR solutions are reported in our previous paper~\cite{CVSANS_PR}. For the CV-SANS measurements of PR solutions, we used PR consisting of hydrogenated (h-) PEG or deuterated (d-) PEG as a liner polymer chain and \(\alpha\)-cyclodextrins (CDs) as rings (Fig.~\ref{fig:Fig3_PR} (a)). The scattering length densities \(\rho\) of h-PEG, d-PEG, and CD were \(0.65 \times 10^{6}\), \(7.1 \times 10^{6}\), and \(2.0 \times 10^{6}\) $\rm{\AA}^{-2}$, respectively. 
h-PR and d-PR were dissolved in mixtures of hydrogenated dimethyl sulfoxide (h-DMSO) and deuterated DMSO (d-DMSO). The volume fraction of PR in the solutions was 8\%. The volume fractions of d-DMSO in the solvent, \(\phi_{d}\) were 1.0, 0.95, 0.90, and 0.85, and the corresponding scattering length densities of the solvents were \(5.3 \times 10^{6}\), \(5.0 \times 10^{6}\), \(4.7 \times 10^{6}\), and \(4.5 \times 10^{6}\) $\rm{\AA}^{-2}$, respectively.
Based on the d-DMSO fraction and type of PR, the PR solutions are named as hPR100, hPR095, hPR090, hPR085, dPR100, dPR095, dPR090, dPR085, as shown in Fig.~\ref{fig:Fig3_PR} (b).

\begin{figure}
    \centering
    \includegraphics[width=1.0\linewidth]{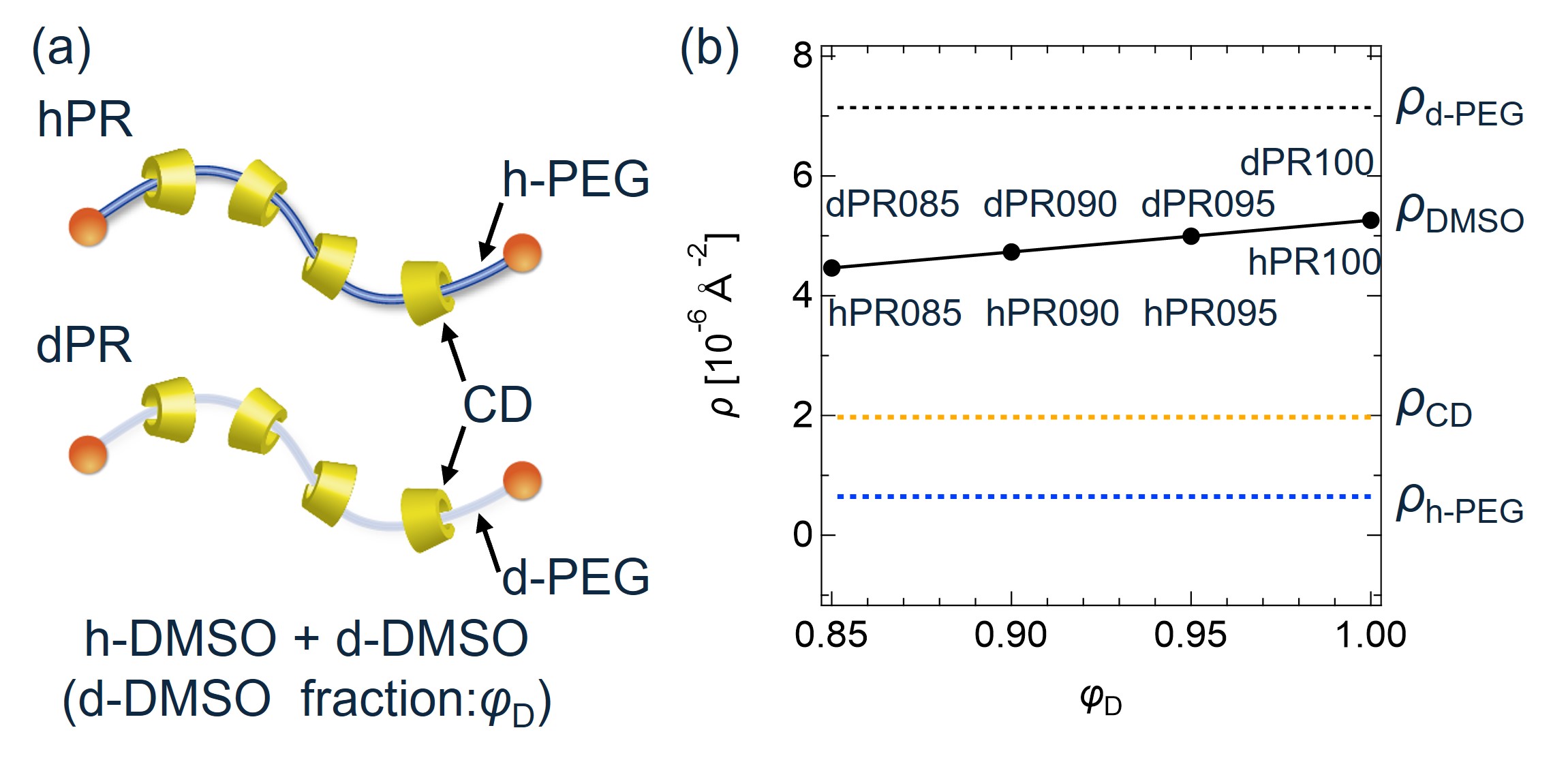}
    \caption{(a) Schematic illustration of a polyrotaxane (PR) solution dissolved in a d-DMSO/h-DMSO mixture; (b) scattering length densities of h-PEG (\(\rho \rm{_{h-PEG}}\)), d-PEG(\(\rho\rm{_{d-PEG}}\)), CD(\(\rho\rm{_{CD}}\)), and solvent(\(\rho \rm{_{DMSO}}\)) plotted against d-DMSO fraction of the solvent, \(\phi_{\rm{D}}\). }
    \label{fig:Fig3_PR}
\end{figure}

The SANS measurements of the clay/PEG and PR solutions were performed at 298 K by using the SANS-U diffractometer of the Institute for Solid State Physics, The University of Tokyo, located at the JRR-3 research reactor of the Japan Atomic Energy Agency in Tokai, Japan. The incident beam wavelength was 7.0 \AA, and the wavelength distribution was 10\%. Sample-to-detector distances were 1 and 8 m for the clay/PEG solutions, and 1 and 4 m for the PR solutions, respectively. The scattered neutrons were collected with a two-dimensional detector and then the necessary corrections were made, such as air and cell scattering subtractions. After these corrections, the scattered intensity was normalized to the absolute intensity using a standard polyethylene film with known absolute scattering intensity. The two-dimensional intensity data were circularly averaged and the incoherent scattering was subtracted.
The averaged scattering intensities, \(I\), were plotted against the magnitude of the scattering vector, \( Q \). The error bars of \(I(Q)\) were given by \(\Delta I\) = \(\pm \sigma\), where \(\sigma\) represents the standard deviation of the circular averaging.

\section{Results}

\subsection{Error estimation for computational data of core-shell sphere}\label{sec:results_cs}

The computed scattering intensities, \(I(Q)\), of the core-shell sphere are shown in Fig.~\ref{fig:Fig4_I_cs}. 
The relative error of \(I(Q)\), \(\Delta I(Q) / I(Q)\), is set at \(\pm\) 0.05, giving the error bars in Fig.~\ref{fig:Fig4_I_cs}.

The scattering intensity, $I(Q)$, of the core-shell sphere is represented as follows:
\begin{equation}\label{eq:I_S_cs}
    I(Q) = \Delta\rho_{\rm C}^2 S_{\rm CC}(Q) + \Delta\rho_{\rm S}^2 S_{\rm SS}(Q) + 2\Delta\rho_{\rm C}\Delta\rho_{\rm S} S_{\rm CS}(Q)
\end{equation}
Here, \( S_{\rm CC}(Q) \) is the self-term of the core, \( S_{\rm SS}(Q) \) is the self-term of the shell, and \( S_{\rm CS}(Q) \) is the cross-term between the core and shell.
This section considers the case when $S_{11}$, $S_{22}$, $S_{12}$, $\Delta_i\rho_1$, and $\Delta_i\rho_2$ for $i = 1,\dots,N$ in Eq. \eqref{eq:LSE} correspond to $S_{\rm CC}$, $S_{\rm SS}$, $S_{\rm CS}$, $\Delta\rho_{\rm C}$, and $\Delta\rho_{\rm S}$, respectively. 

\begin{figure}
    \centering
    \includegraphics[width=0.8\linewidth]{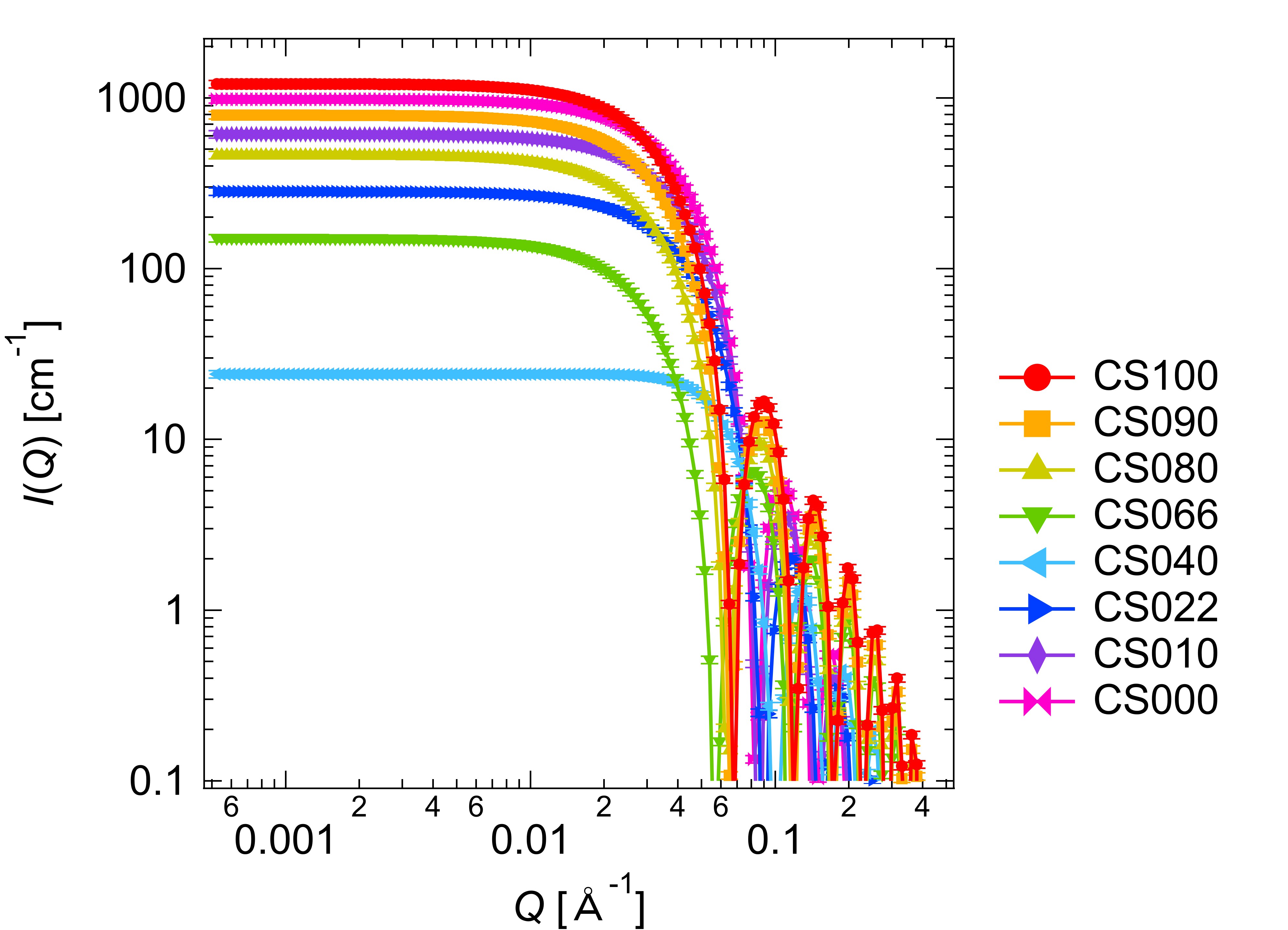}
    \caption{Computed scattering intensities, \(I(Q)\), with error bars of the core-shell sphere: CS100, CS090, CS080, CS066, CS040, CS022, CS010, and CS000. The error bars are given by \(\Delta I(Q) / I(Q) = \pm 0.05\)}
    \label{fig:Fig4_I_cs}
\end{figure}

\begin{enumerate}
  \item[\textit{3.1.1}] \textit{Deterministic error estimation of core-shell sphere}
\end{enumerate}

In this section, we present the results of the deterministic error estimation described in Section~\ref{sec:DEE}. 
Let $I$, $\Delta I$, $A$, $S$, and $\Delta S$ be as defined in Eq.(\ref{eq:notation-IAS}).  
Here, the vector $\Delta J$ in Theorem~\ref{th:AE} is 0.05 $\times I$. 

Using various combinations of three scattering intensities out of the eight data shown in Fig.~\ref{fig:Fig4_I_cs}, we calculated the partial scattering functions and their errors for the core-shell sphere. Denote the upper bounds on $\abs{\Delta S_{\rm CC}}$, $\abs{\Delta S_{\rm SS}}$, and $\abs{\Delta S_{\rm CS}}$ obtained based on Theorem~\ref{th:AE} by $\Delta T_{\rm CC}$, $\Delta T_{\rm SS}$, and $\Delta T_{\rm CS}$, respectively. 
Figs.~\ref{fig:Fig5_cs_miyajjima} (a)-(f) display the numerically computed partial scattering functions, $S_{\rm CC} + \Delta S_{\rm CC}$, $S_{\rm SS} + \Delta S_{\rm SS}$, and $S_{\rm CS} + \Delta S_{\rm CS}$. The error bars for the partial scattering functions are given by $\Delta T_{\rm CC}$, $\Delta T_{\rm SS}$, and $\Delta T_{\rm CS}$. 
Fig.~\ref{fig:Fig5_cs_miyajjima} also shows the relative errors of the partial scattering functions, defined as  $\Delta T_{\rm CC}/(S_{\rm CC} + \Delta S_{\rm CC})$, $\Delta T_{\rm SS}/(S_{\rm SS} + \Delta S_{\rm SS})$, and $\Delta T_{\rm CS}/(S_{\rm CS} + \Delta S_{\rm CS})$. 

The calculated partial scattering functions are identical regardless of contrast combinations. Additionally, note that \( S_{\rm CC}(Q) \) and \( S_{\rm SS}(Q) \) completely overlap with $I(Q)/\Delta\rho_{\rm C}^2$ for CS022 (shell matching, black solid lines in Fig.~\ref{fig:Fig5_cs_miyajjima}) and $I(Q)/\Delta\rho_{\rm S}^2$ for CS066 (core matching, black dotted lines in Fig.~\ref{fig:Fig5_cs_miyajjima}), respectively, indicating the validity of our calculation results. The relative errors of $S$ vary with different contrast combinations corresponding to various condition numbers of the matrix $A$, cond(A), as shown in Fig.~\ref{fig:Fig5_cs_miyajjima}.

\begin{figure}
    \centering
    \includegraphics[width=0.9\linewidth]{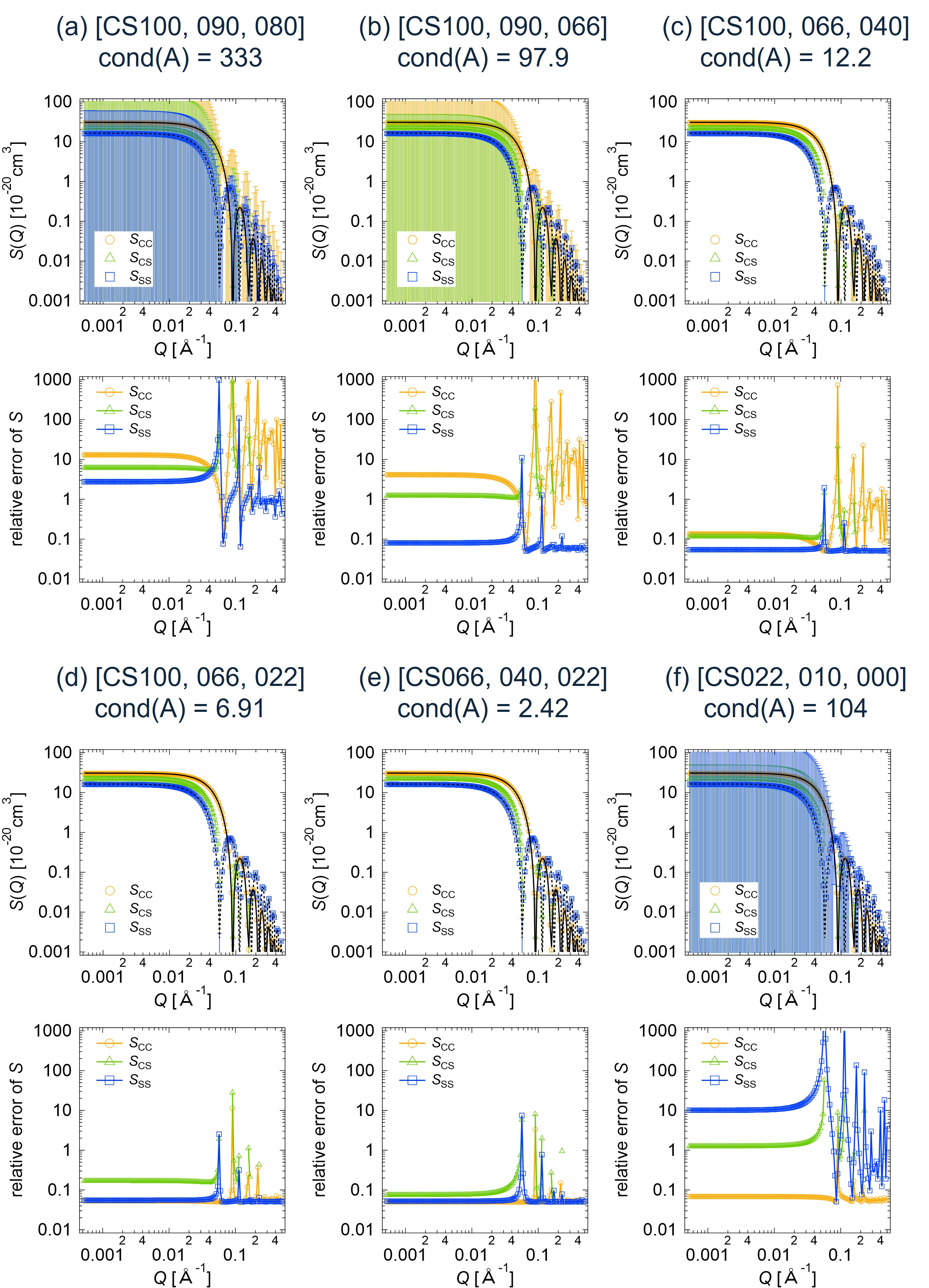}
    \caption{Partial scattering functions with error bars and their relative errors for the core-shell sphere, obtained by applying the deterministic error estimation method to the different data set of the computed scattering intensities: (a) CS100, CS090, CS080; (b) CS100, CS090, CS066; (c) CS100, CS066, CS040; (d) CS100, CS066, CS022; (e) CS066, CS040, CS022; (f) CS022, CS011, CS000. cond(A) is the condition number of the matrix A. The black solid and dotted lines correspond to $I(Q)/\Delta\rho_{\rm C}^2$ of CS022 (shell matching) and $I(Q)/\Delta\rho_{\rm S}^2$ of CS066 (core matching), respectively.}
    \label{fig:Fig5_cs_miyajjima}
\end{figure}


\begin{enumerate}
  \item[\textit{3.1.2}] \textit{Statistical error estimation of core-shell sphere}
\end{enumerate}

Next, we applied the statistical method described in Section~\ref{sec:statistical} to the computational core-shell sphere data by setting $I_1, \ldots, I_N$ to the computed scattering intensities and $\sigma_1, \ldots, \sigma_N$ to the artificial errors, 0.05$\times I$.
Fig.~\ref{fig:Fig6_cs_ohbayashi} shows the estimated partial scattering functions and their errors computed by Eq.~\eqref{eq:posterior-1} for the different contrast combinations (a) to (f). 
The error bars represent $\bar{S}_{\rm CC} \pm \bar{\sigma}_{\rm CC},\  \bar{S}_{\rm SS} \pm \bar{\sigma}_{\rm SS}$ and $\bar{S}_{\rm CS} \pm \bar{\sigma}_{\rm CS}$, where $\bar{\sigma}_{\rm CC} = \sqrt{\bar{\sigma}_{\rm CC,\rm CC}}$, $\bar{\sigma}_{\rm SS} = \sqrt{\bar{\sigma}_{\rm SS,\rm SS}}$, and $\bar{\sigma}_{\rm CS} = \sqrt{\bar{\sigma}_{\rm CS,\rm CS}}$. Similar to the results of the deterministic estimation, the calculated partial scattering functions are the same for all the contrast combinations, and the obtained \( S_{\rm CC}(Q) \) and \( S_{\rm SS}(Q) \) are identical with $I(Q)/\Delta\rho_{\rm C}^2$ of CS022 (shell matching, black solid lines in Fig. ~\ref{fig:Fig6_cs_ohbayashi}) and $I(Q)/\Delta\rho_{\rm S}^2$ of CS066 (core matching, black dotted lines in Fig. ~\ref{fig:Fig6_cs_ohbayashi}).
Fig.~\ref{fig:Fig6_cs_ohbayashi} also shows the relative estimated errors of the partial scattering functions, namely $\bar{\sigma}_{\rm CC} / \bar{S}_{\rm CC}, \bar{\sigma}_{\rm pp} / \bar{S}_{\rm SS}$, and $\bar{\sigma}_{\rm CS} / \bar{S}_{\rm CS}$. 

\begin{figure}
    \centering
    \includegraphics[width=0.9\linewidth]{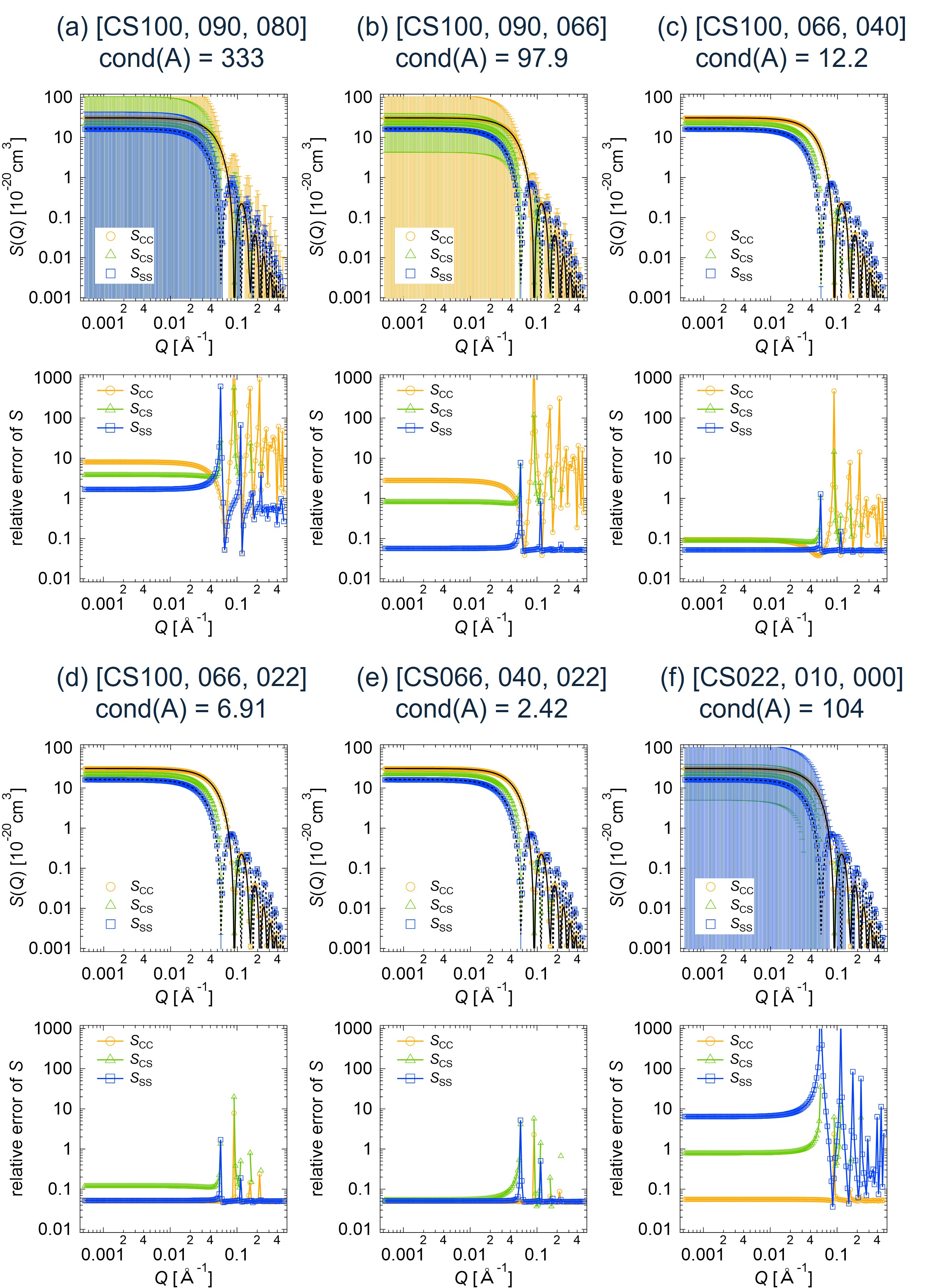}
    \caption{Partial scattering functions with error bars and their relative errors for the core-shell sphere, obtained by applying the statistical error estimation method to the different data sets of the computed scattering intensities: (a) CS100, CS090, CS080; (b) CS100, CS090, CS066; (c) CS100, CS066, CS040; (d) CS100, CS066, CS022; (e) CS066, CS040, CS022; (f) CS022, CS011, CS000.The black solid and dotted lines correspond to  $I(Q)/\Delta\rho_{\rm C}^2$ of CS022 (shell matching) and $I(Q)/\Delta\rho_{\rm S}^2$ of CS066 (core matching), respectively.}
    \label{fig:Fig6_cs_ohbayashi}
\end{figure}

\begin{enumerate}
  \item[\textit{3.1.3}] \textit{Comparison between deterministic and statistical error estimation results of core-shell sphere}
\end{enumerate}

We compare the results obtained from the determinisic error estimation (Fig.~\ref{fig:Fig5_cs_miyajjima}) and the statistical error estimation (Fig.~\ref{fig:Fig6_cs_ohbayashi}). 
In Fig.~\ref{fig:Fig7_relS_ef_cs}, the relative errors of the partial scattering functions $S$ at $Q = 0.01~\rm{\AA}^{-1}$ are plotted against the condition number of $A$. The statistical estimation yields smaller relative errors of $S$ compared to the deterministic estimation.
This difference in error propagation may result from the different assumptions underlying the two methods.
The deterministic error estimation requires weaker assumptions than the statistical method (see Sections~\ref{sec:DEE} and ~\ref{sec:statistical}), 
which results in the larger relative errors of $S$ obtained from the deterministic method.

Furthermore, Fig.~\ref{fig:Fig7_relS_ef_cs} shows a positive correlation between the relative error of $S$ and the condition number of $A$. This suggests that increasing the condition number of $A$ results in increasing the degree of error propagation from the scattering intensities to the partial scattering functions, which is consistent with the explanation of the condition number shown in Appendix \ref{sec:condnum}. We define an error propagation factor as the relative error of $S$ divided by the relative error of $I$, 0.05. The right axis of Fig.~\ref{fig:Fig7_relS_ef_cs} represents this error propagation factor. When the condition number of $A$ is at its minimum value, 2.42, the error propagation factors for all the partial scattering functions are close to 1, indicating that the error propagation is minimized and that the partial scattering functions are accurately determined. However, for the maximum condition number of $A$, 333, the error propagation factors range from 30 to 200, resulting in the large relative errors of $S$ exceeding 1 and the large error bars of $S$ shown in Figs.~\ref{fig:Fig5_cs_miyajjima} (a) and ~\ref{fig:Fig6_cs_ohbayashi} (a).

\begin{figure}
    \centering
    \includegraphics[width=0.8\linewidth]{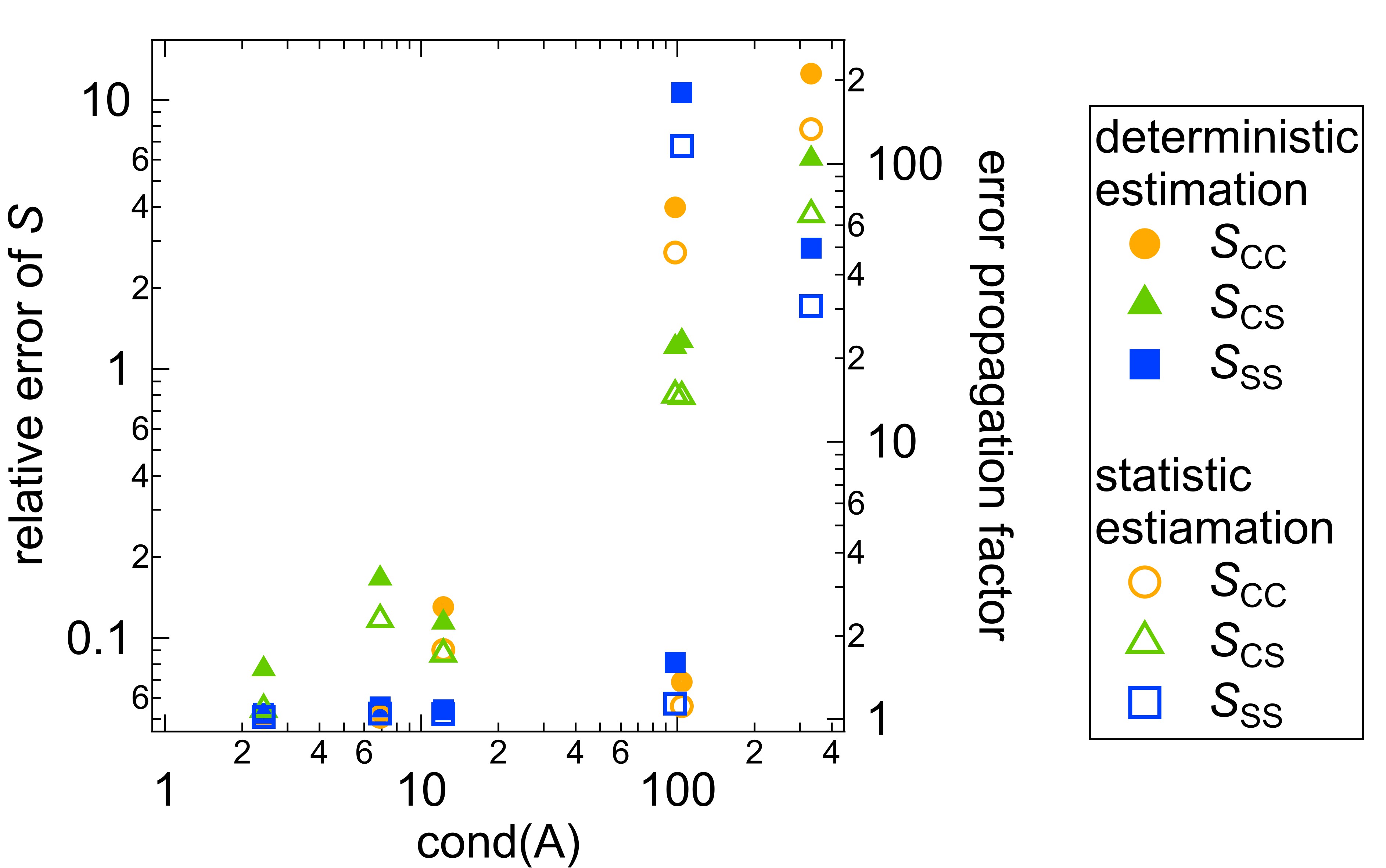}
    \caption{Plot of the relative errors of $S$ versus the condition number of $A$ for the core-shell sphere system. The right axis represents the error propagation factor.}
    \label{fig:Fig7_relS_ef_cs}
\end{figure}

\subsection{Error estimation for experimental data of clay/PEG solutions}\label{sec:results_cp}

The experimentally measured scattering intensities  \(I(Q)\) of the clay/PEG solutions with various D\(_{2}\)O fractions are shown in Fig.~\ref{fig:Fig8_I_cp} (a). As described in Section~\ref{subsec:method-cp-PR}, the error bars for \(I(Q)\) are given by \(\Delta I\) = \(\pm \sigma\), in which \(\sigma\) is the standard deviation of the circular averaging. The relative errors of \(I(Q)\), \(\sigma/I(Q)\), are shown in Fig.~\ref{fig:Fig8_I_cp} (b).

The scattering intensities $I(Q)$ of the clay/PEG solutions are given by the following equation:
\begin{equation}\label{eq:I_S_cp}
    I(Q) = \Delta\rho_{\rm C}^2 S_{\rm CC}(Q) + \Delta\rho_{\rm P}^2 S_{\rm PP}(Q) + 2\Delta\rho_{\rm C}\Delta\rho_{\rm P} S_{\rm CP}(Q)
\end{equation}
Here, \( S_{\rm CC}(Q) \) is the self-term of the clay particles, \( S_{\rm PP}(Q) \) is the self-term of PEG, and \( S_{\rm CP}(Q) \) is the cross-term between the clay and PEG.
In this section, $S_{11}$, $S_{22}$, $S_{12}$, $\Delta_i\rho_1$, and $\Delta_i\rho_2$ for $i = 1,\dots,N$ in Eq. \eqref{eq:LSE} are represented as $S_{\rm CC}$, $S_{\rm PP}$, $S_{\rm CP}$, $\Delta\rho_{\rm C}$, and $\Delta\rho_{\rm P}$, respectively. 

\begin{figure}
    \centering
    \includegraphics[width=0.7\linewidth]{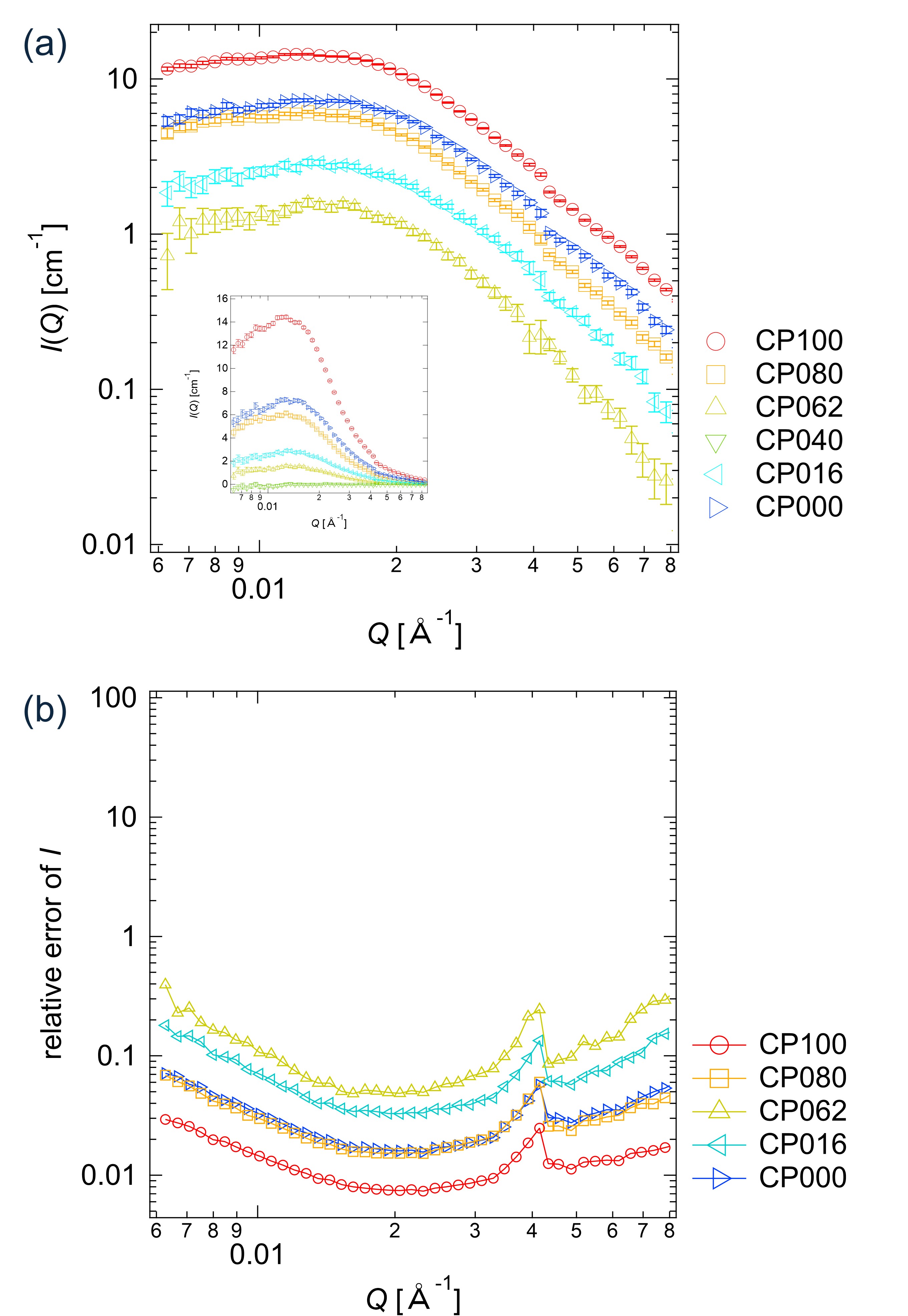}
    \caption{(a) Experimentally measured scattering intensities  \(I(Q)\)  with error bars and (b) relative errors of  \(I(Q)\)'s for the clay/PEG solutions: CP100, CP080, CP062, CP040, CP016 and CP000. In the insert of (a), linear  \(I(Q)\) is plotted against log \(Q\)
    }
    \label{fig:Fig8_I_cp}
\end{figure}

\begin{enumerate}
  \item[\textit{3.2.1}] \textit{Deterministic error estimation of clay/PEG solutions}
\end{enumerate}
 We applied the deterministic error estimation method to the CV-SANS experimental data of the clay/PEG solutions. 
For $c \in \V{n}$ and $r \in \Vp{n}$, define $\itv{c}{r} := [c-r,c+r]$. 
Let $I$, $\Delta I$, $A$, $S$, and $\Delta S$ be as defined in Eq.(\ref{eq:notation-IAS}).  
As mentioned in Remark~\ref{rm:AE}, we considered the standard deviation of the experimentally obtained $I(Q)$ as $\Delta J$ in Theorem~\ref{th:AE}. 
Consequently, the interval $\itv{I + \Delta I}{\Delta J}$ contains $I$ with a probability of approximately 68.3\% if $I$ follows a normal distribution. 
Because $\Delta T$ in Theorem~\ref{th:AE} represents the upper bound on $\abs{\Delta S}$, the interval $\itv{S + \Delta S}{\Delta T}$ contains $S$ with a probability equal to or greater than 68.3\% in this case.  
If $I \in \itv{I + \Delta I}{\Delta J}$ holds rigorously, then $S \in \itv{S + \Delta S}{\Delta T}$. 

In the same way as the deterministic estimation for the core-shell sphere (Section 3.1.1), we calculated the partial scattering functions and their errors for the clay/PEG solutions from various combinations of three data sets among the six with different scattering contrasts. 
Figs.~\ref{fig:Fig9_cs_miyajjima}(a)-(d) show the calculated partial scattering functions and their relative errors for different contrast combinations, corresponding to the condition numbers of $A$ from 2.96 to 48.6. The obtained partial scattering functions are almost the same for all the cases, while the relative errors of the partial scattering functions vary depending on the contrast combination. The cross-term \( S_{\rm CP}(Q) \) is positive, indicating the PEG chains are adsorbed onto the clay particles ~\cite{clayPEG}.

\begin{figure}
    \centering
    \includegraphics[width=1.0\linewidth]{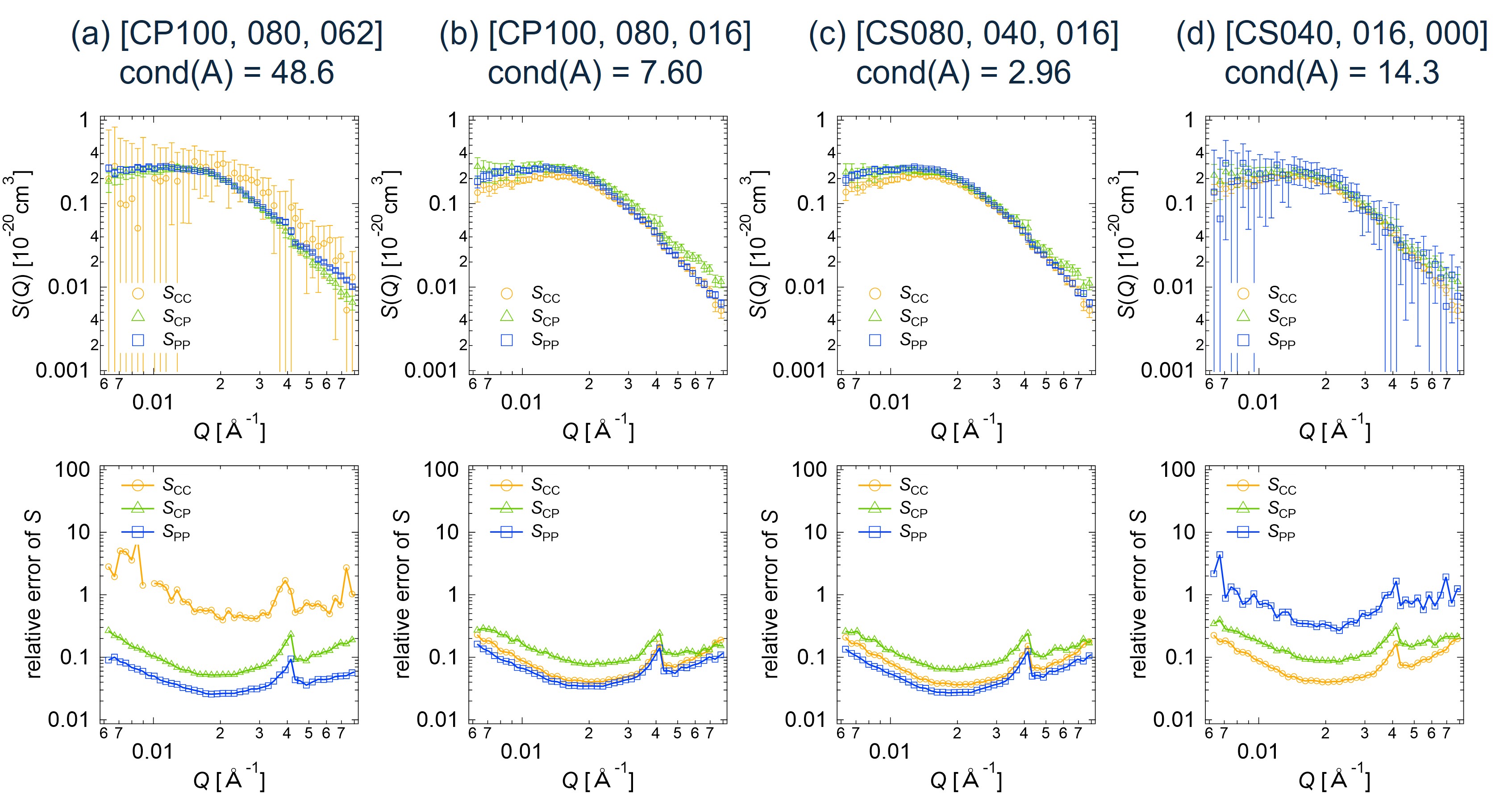}
    \caption{Partial scattering functions with error bars and their relative errors for the clay/PEG solutions, obtained by applying the deterministic error estimation method to the different data sets of the SANS scattering intensities: (a) CP100, CP080, CP062; (b) CP100, CP080, CP016; (c) CP080, CP040, CP016; (d) CP040, CP016, CP000.}
    \label{fig:Fig9_cs_miyajjima}
\end{figure}

\begin{enumerate}
  \item[\textit{3.2.2}] \textit{Statistical error estimation of clay/PEG solutions}
\end{enumerate}

Here, we present the statistical error estimation results for the clay/PEG solutions.
For this analysis, we set $I_1, \ldots, I_N$ to the circularly averaged scattering intensities and $\sigma_1, \ldots, \sigma_N$ to the standard deviations of these averages.
Fig.~\ref{fig:Fig10_cp_ohbayashi} shows the partial scattering functions and their relative errors, computed using Eq.~\eqref{eq:posterior-1} for the different contrast combinations (a) to (d). The partial scattering functions obtained through the statistical method are quite similar to those obtained using the deterministic method, which is shown in Fig.~\ref{fig:Fig9_cs_miyajjima}.

\begin{figure}
    \centering
    \includegraphics[width=1.0\linewidth]{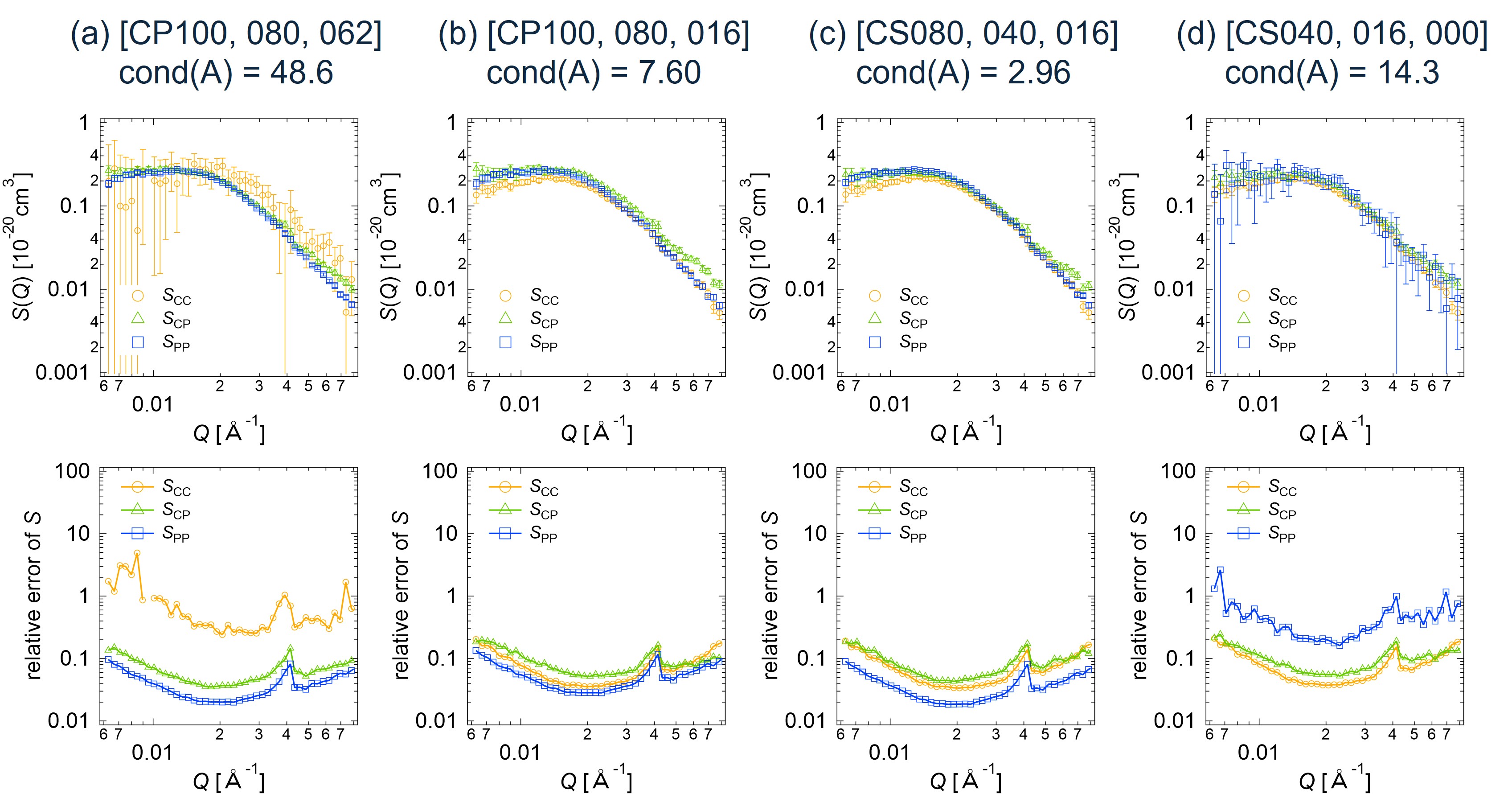}
    \caption{Partial scattering functions with error bars and their relative errors for the clay/PEG solutions, obtained by applying the statistical error estimation method to the different data sets of the SANS scattering intensities: (a) CP100, CP080, CP062; (b) CP100, CP080, CP016; (c) CP080, CP040, CP016; (d) CP040, CP016, CP000.}
    \label{fig:Fig10_cp_ohbayashi}
\end{figure}

\begin{enumerate}
  \item[\textit{3.2.3}] \textit{Comparison between deterministic and statistical error estimation results for clay/PEG solutions}
\end{enumerate}

Fig.~\ref{fig:Fig11_relS_cp} displays the relationship between the condition number of $A$ and the relative errors of the partial scattering functions $S$ for the clay/PEG solutions at $Q = 0.02~\rm{\AA}^{-1}$. This is similar to that observed for the core-shell sphere (Fig.~\ref{fig:Fig7_relS_ef_cs}). 
The relative errors calculated by the statistical method are smaller than those obtained with the deterministic estimation method. 
Additionally, reducing the condition number of $A$ decreases the relative errors of $S$. When the condition number of $A$ is 2.96 or 7.60, the relative errors of $S$ are less than 0.1, and all the three partial scattering functions are determined with high accuracy. In contrast, for cond($A$) = 14.3 or 48.6, the relative error of at least one partial scattering function exceeds 0.2, resulting in the large error bars in Figs.~\ref{fig:Fig9_cs_miyajjima} (a)(d) and Figs.~\ref{fig:Fig10_cp_ohbayashi} (a)(d).

\begin{figure}
    \centering
    \includegraphics[width=0.8\linewidth]{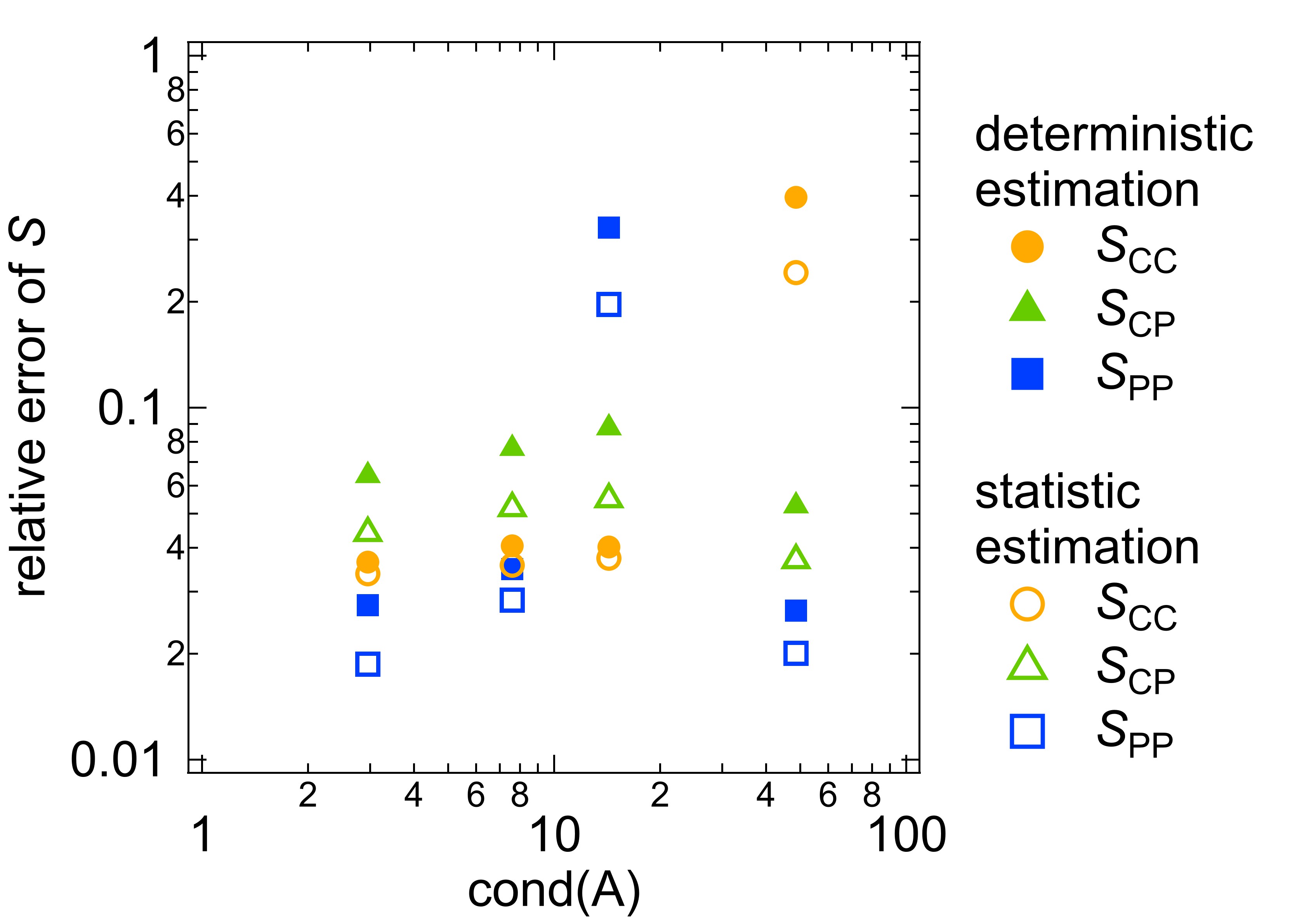}
    \caption{Plot of the relative errors of $S$ versus the condition number of $A$ for the clay/PEG solutions.}
    \label{fig:Fig11_relS_cp}
\end{figure}

\subsection{Error estimation for experimental data of PR solutions}\label{sec:results_PR}

Figs.~\ref{fig:Fig12_I_PR} (a) and (b) show the scattering intensities  \(I(Q)\) and relative errors of \(I(Q)\), \(\sigma/I(Q)\), for the PR solutions with different scattering contrasts~\cite{CVSANS_PR}. The scattering intensities $I(Q)'s$ of the polyrotaxane(PR)  solutions are described by the following equation:
\begin{equation}\label{eq:I_S_PR}
    I(Q) = \Delta\rho_{\rm C}^2 S_{\rm CC}(Q) + \Delta\rho_{\rm P}^2 S_{\rm PP}(Q) + 2\Delta\rho_{\rm C}\Delta\rho_{\rm P} S_{\rm CP}(Q)
\end{equation}
where, \( S_{\rm CC}(Q) \) is the self-term for CDs, \( S_{\rm PP}(Q) \) is the self-term for PEG, and \( S_{\rm CP}(Q) \) is the cross-term between CD and PEG.
For the PR solutions, $S_{11}$, $S_{22}$, $S_{12}$, $\Delta_i\rho_1$, and $\Delta_i\rho_2$ for $i = 1,\dots,N$ in Eq. \eqref{eq:LSE} correspond to  $S_{\rm CC}$, $S_{\rm PP}$, $S_{\rm CP}$, $\Delta\rho_{\rm C}$, and $\Delta\rho_{\rm P}$, respectively. 

\begin{figure}
    \centering
    \includegraphics[width=0.7\linewidth]{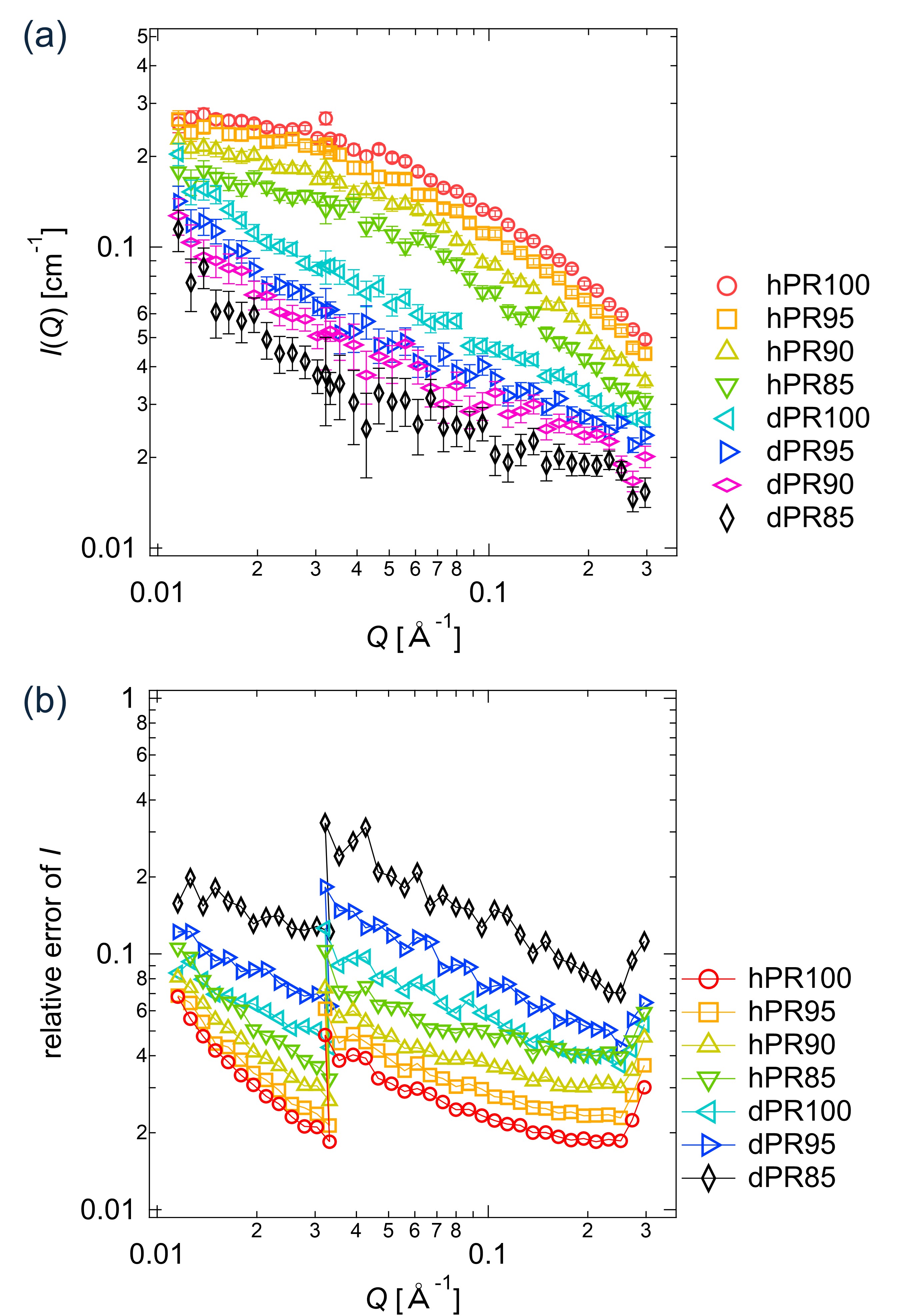}
    \caption{(a) Experimentally measured scattering intensities  \(I(Q)\) with error bars and (b) relative errors of \(I(Q)\) for the PR solutions: hPR100, hPR095, hPR090, hPR085, dPR100, dPR095, dPR090, and dPR085.
    }
    \label{fig:Fig12_I_PR}
\end{figure}

In the same manner as for the CV-SANS data of the clay/PEG solutions, we performed the deterministic and statistical error estimations for the PR solutions. Figs.~\ref{fig:Fig13_PR_miyajjima} and \ref{fig:Fig14_PR_ohbayashi} show the partial scattering functions and their relative errors for the PR solutions using the deterministic and statistical methods, respectively. For the case (a) in Figs.~\ref{fig:Fig13_PR_miyajjima} and \ref{fig:Fig14_PR_ohbayashi} , we used all eight SANS data (four contrasts of h-PR solutions and four contrasts of d-PR solutions). In this case, \( S_{\rm CC}(Q) \) and \( S_{\rm CP}(Q) \) are determined with sufficient accuracy. The positive cross-term $S_{\rm cp}$ represents the topological connection between CD and PEG ~\cite{CVSANS_PR, model_PR}. $S_{\rm cc}$, corresponding to the alignment of CDs on PEG, can be described by a random copolymer model~\cite{CVSANS_PR, model_PR}. However, the relative error of \( S_{\rm PP}(Q) \) is greater than 1, making it difficult to discuss the structure of PEG in PR based on \( S_{\rm PP}(Q) \).

For cases of (b) and (d), where only 4 data of h-PR or d-PR were used, the relative errors of all partial scattering functions exceed 1, indicating that both of h-PR and d-PR data are necessary to reduce the error propagation. As shown in Figs.~\ref{fig:Fig13_PR_miyajjima} (c) and \ref{fig:Fig14_PR_ohbayashi} (c), when two data of h-PR solutions (hPR100 and hPR085) and two data of d-PR solutions (dPR100 and dPR085) were analyzed, the relative errors of the partial scattering functions are almost the same as those obtained from the eight-contrast case (Figs.~\ref{fig:Fig13_PR_miyajjima} (a) and \ref{fig:Fig14_PR_ohbayashi} (a)). 

\begin{figure}
    \centering
    \includegraphics[width=1.0\linewidth]{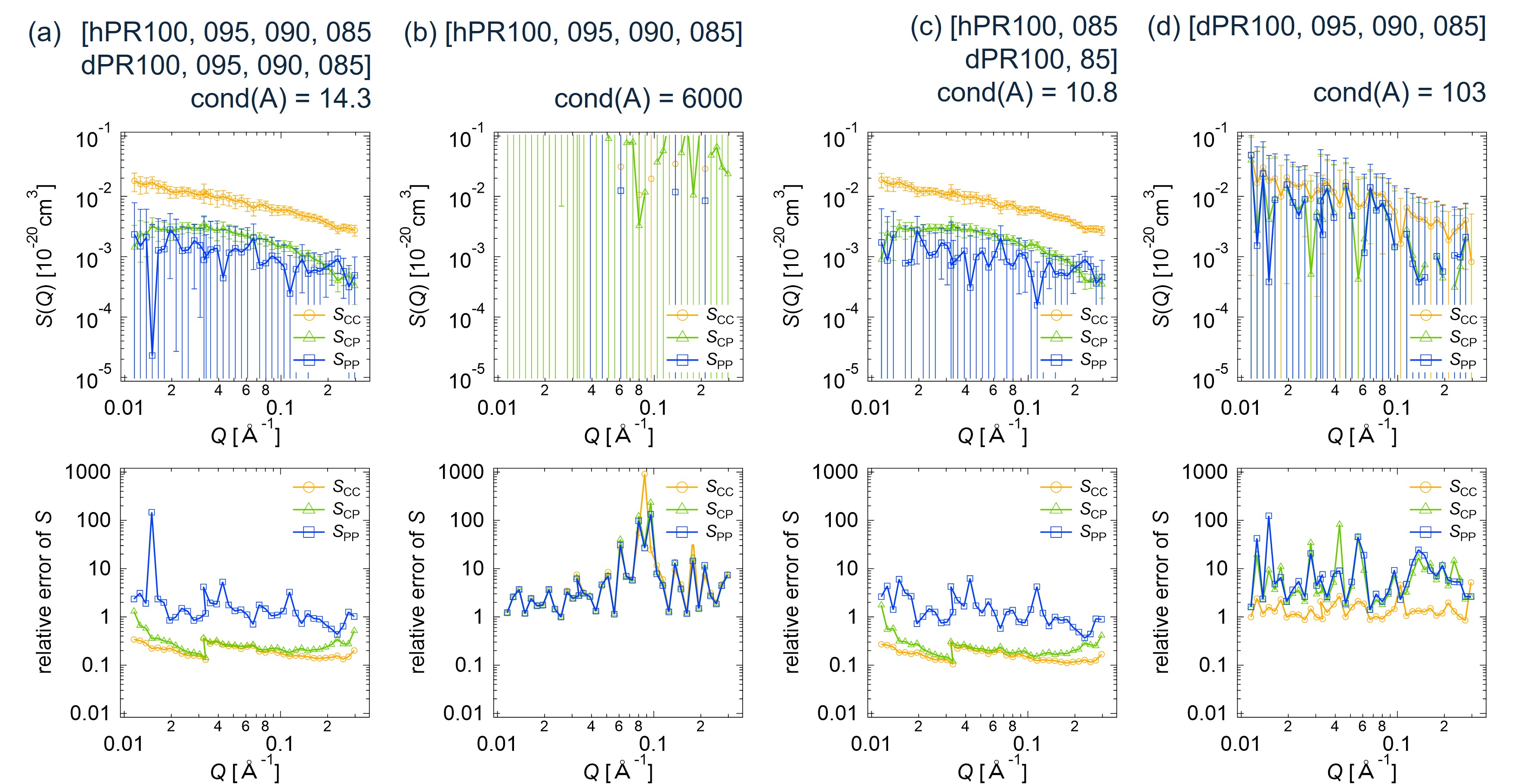}
    \caption{Partial scattering functions with error bars and their relative errors for the PR solutions, obtained by applying the deterministic error estimation method to the different data sets of the SANS scattering intensities: (a) hPR100, hPR095, hPR090, hPR085, dPR100, dPR095, dPR090, dPR085; (b) hPR100, hPR095, hPR090, hPR085; (c) hPR100, hPR085, dPR100, dPR085; (d) dPR100, dPR095, dPR090, dPR085.}
    \label{fig:Fig13_PR_miyajjima}
\end{figure}

\begin{figure}
    \centering
    \includegraphics[width=1.0\linewidth]{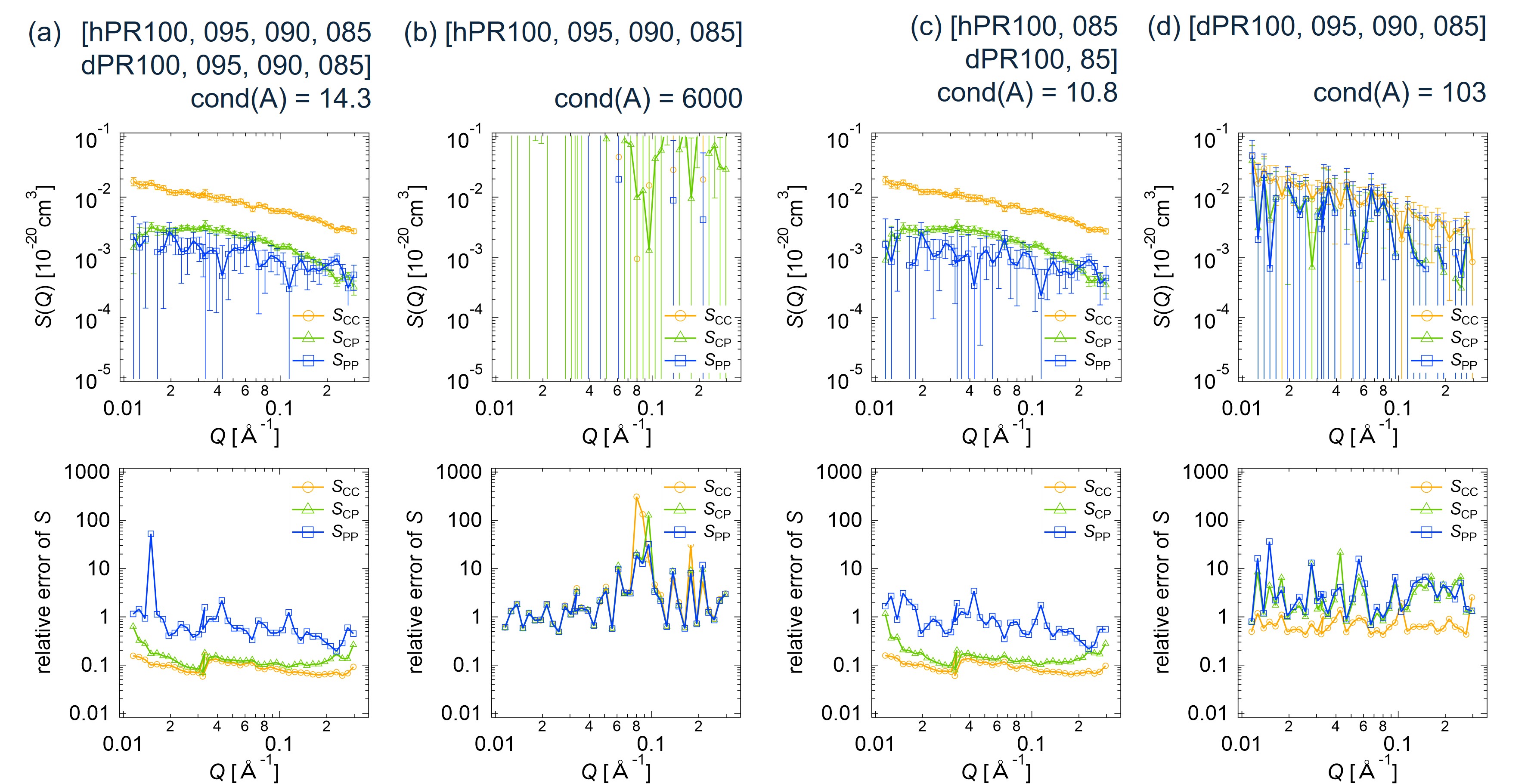}
    \caption{Partial scattering functions with error bars and their relative errors for the PR solutions, obtained by applying the statistical error estimation method to the different data sets of the SANS scattering intensities: (a) hPR100, hPR095, hPR090, hPR085, dPR100, dPR095, dPR090, dPR085; (b) hPR100, hPR095, hPR090, hPR085; (c) hPR100, hPR085, dPR100, dPR085; (d) dPR100, dPR095, dPR090, dPR085.}
    \label{fig:Fig14_PR_ohbayashi}
\end{figure}

In Fig.~\ref{fig:Fig15_relS_PR}, the relative errors of the partial scattering functions $S$ of the PR solutions at $Q = 0.02~\rm{\AA}^{-1}$ are plotted against the condition numbers of $A$ for the different contrast combinations. 
In the same manner as for the core-shell sphere systems and clay/PEG solutions, minimizing the condition number of $A$ is important for more precise determination of the partial scattering functions. The cases (a) and (c) in Figs.~\ref{fig:Fig13_PR_miyajjima} and \ref{fig:Fig14_PR_ohbayashi} correspond to small condition numbers around 10, resulting in the similar error estimation results for the two cases. 

\begin{figure}
    \centering
    \includegraphics[width=0.8\linewidth]{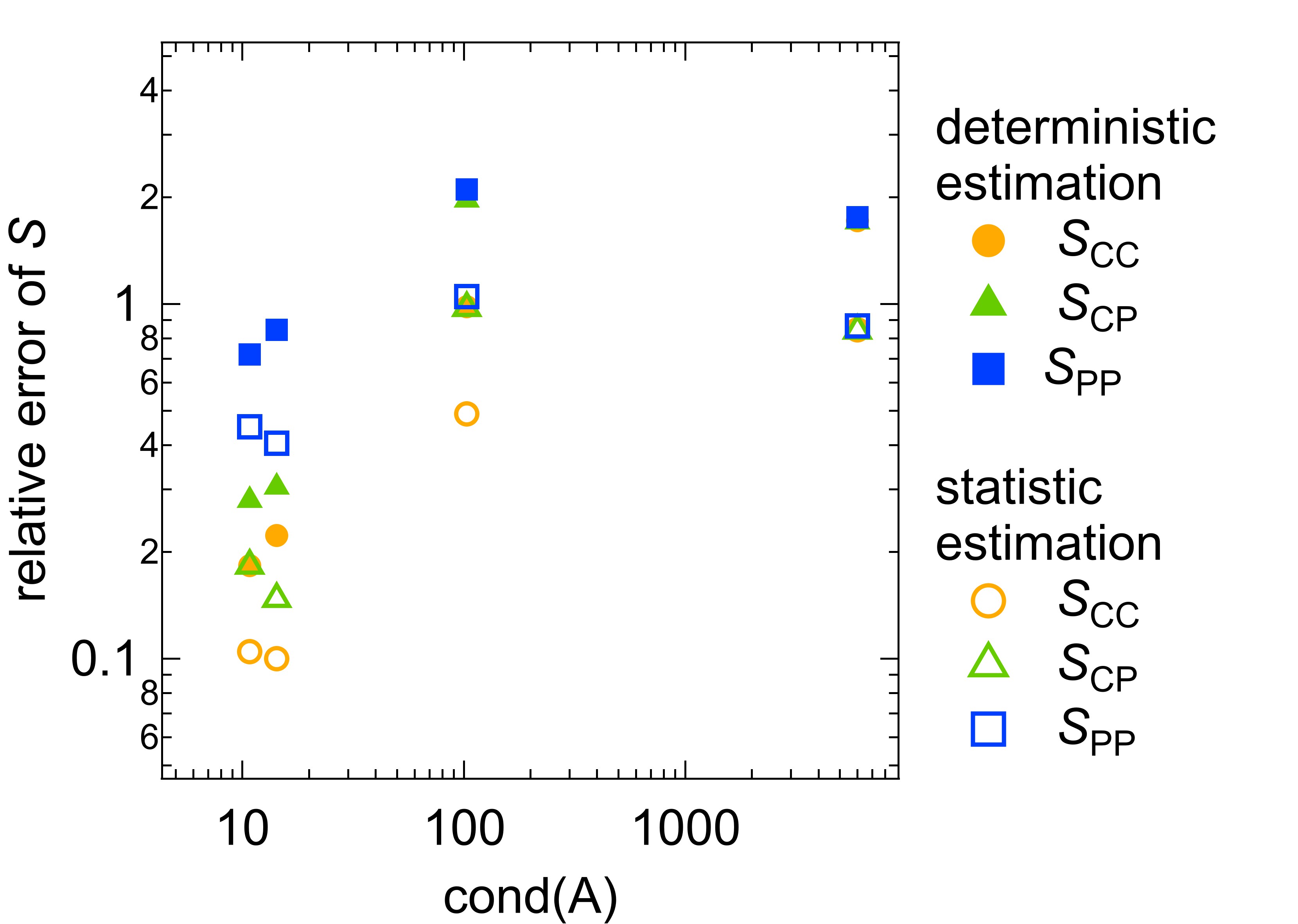}
    \caption{Plot of the relative errors of $S$ versus the condition number of $A$ for the PR solutions.}
    \label{fig:Fig15_relS_PR}
\end{figure}

\section{Conclusion}

In this study, we have established the deterministic and statistical error estimation methods for calculating partial scattering functions from scattering intensities of CV-SANS data. By applying these methods to (i) computational data of a core-shell sphere and experimental CV-SANS data of (ii) clay/polyethylene glycol (PEG) aqueous solutions and (iii) polyrotaxane solutions, we successfully achieved theoretically grounded error estimations of their partial scattering functions.
This approach is valuable for evaluating the reliability of partial scattering functions computed from CV-SANS data. 
The statistical error estimation requires more assumptions than the deterministic error estimation, but the former usually gives sharper results than the latter. Therefore, statistical error estimation is better if the assumptions are valid; otherwise, deterministic is better.

This study also highlighted the significance of the singular values of the matrix $A$ appearing on the right side of the problem~\eqref{SVSANSmodel} in predicting error bars of partial scattering functions.
For both the deterministic and the statistical methods, the inverse of the minimum singular value, $1/ \lambda_{\min}$, provides the scaling factor of absolute errors from CV-SANS measurements to the partial scattering functions, while the condition number, $\lambda_{\max}/\lambda_{\min}$, offers the scaling factor of relative errors.
Because the condition number of $A$ can be calculated only from scattering length densities $\rho_i$, the scaling factor can be estimated before CV-SANS measurements.
Therefore, by minimizing the condition number of $A$, we can optimize the choice of contrasts to reduce the error propagation in the CV-SANS data analysis.

Additionally, the error estimation can be used to reduce the number of samples required and shorten SANS measurement times.

For example, Figs.~\ref{fig:Fig13_PR_miyajjima} and~\ref{fig:Fig14_PR_ohbayashi} demonstrate that only four CV-SANS experimental data (case (c): cond($A$)= 10.8) provide almost the same error bars of \( S_{\rm CC}(Q) \) and \( S_{\rm CP}(Q) \) as all the eight CV-SANS experimental data (case (a): cond($A$)= 14.3); this fact suggests the possibility of reducing experimental costs using condition numbers.

\section{Acknowledgments}
This work is supported by the financial support of the JST FOREST Program (grant number JPMJFR2120) and Data Creation and Utilization-Type Material Research and Development Project grant number JPMXP1122714694. The SANS
experiment was carried out by the JRR-3 general user program managed by the Institute for Solid State Physics, The University of Tokyo (Proposal No. 7607 and No.23559).

\appendix
\section{Proof of Theorems}\label{sec:proofs}

\begin{proof}[Proof of Theorem~\ref{th:AE}]
The assumptions $AS = I$ and $I + \Delta I = A(S + \Delta S)$ imply $A\Delta S = \Delta I$. 
Since $A$ has full column rank, we have $\Delta S = A^+\Delta I$, so $\abs{\Delta I} \le \Delta J$ yields
$$
\abs{\Delta S} = \abs{A^+ \Delta I} \le \abs{A^+}\abs{\Delta I} \le \Delta T. 
$$
\end{proof}

\begin{proof}[Proof of Theorem~\ref{th:Vij}]
Since $|\text{an element of } \bar{\Sigma}| \leq \|\bar{\Sigma}\|$, we will estimate the upper bound of $\|\bar{\Sigma}\|$. By using the minimal singular value, we have the following relation:
\begin{align*}
    \|\bar{\Sigma}\| = \|{(A^T\Sigma^{-1}A)}^{-1}\| = {(\text{minimal singular value of } A^T\Sigma^{-1}A)}^{-1}.
\end{align*}
Since $A^{T}\Sigma^{-1}A$ is symmetric, the minimum singular value can be estimated as follows:
\begin{align*}
    \text{minimum singular value of } A^T\Sigma^{-1}A  &= \min_{\|u\|=1} u^T A^T {(\sqrt{\Sigma})}^{-2} A u  \\
    &=\min_{\|u\|=1} \|{(\sqrt{\Sigma})}^{-1} A u\|^2  \\
    &\geq \min_{\|u\|=1} [ (\min_{i}  \sigma_i^{-2}) \|A u\|^2 ] \\
    &\geq (\min_{i} \sigma_i^{-2}) \lambda_{\min}^2,
\end{align*}
where
\begin{equation*}
    \sqrt{\Sigma} = \begin{bmatrix}
        \sigma_1 &  & \\
        & \ddots & \\
         & & \sigma_N
    \end{bmatrix}.
\end{equation*}
The above inequality provides the desired inequality.

\end{proof}

\section{Condition Numbers of Matrices}\label{sec:condnum}
The condition number of a matrix $A$, denoted by ${\rm cond}(A)$, can be defined in two equivalent ways:
\begin{align*}
{\rm cond}(A) &= \frac{\lambda_{\max}}{\lambda_{\min}} = \|A\|\|A^+\|,
\end{align*}
where $A^{+}$ denotes the Moore-Penrose inverse of $A$, and $\lambda_{\max}$ and $\lambda_{\min}$ denote the maximum and minimum singular values of $A$, respectively.  
For symmetric positive definite matrices, singular values coincide with eigenvalues.
The condition number of $A$ is a measure of the sensitivity of the solution to a linear system involving the coefficient matrix to changes or errors in the input data.
A high condition number indicates that the matrix is close to being singular or ill-conditioned, implying that small changes in the input can lead to large changes in the output.

To understand this intuitively, consider a system $Ax = b$ with a $2\times 2$ symmetric positive definite matrix $A$ and how $A$ transforms vectors.
In Fig. \ref{fig:condition_number}, the square $(0,1)^2 \ni x$ on the left exists in the solution space and is transformed into the parallelogram $Ax$ on the right in the space.
Since $A$ is symmetric positive definite, the singular values $\lambda_{\max}$ and $\lambda_{\min}$ are also eigenvalues. 
The eigenvalues $\lambda_{\max}$ and $\lambda_{\min}$ indicate maximum and minimum stretching, respectively. When $A$ has a large condition number, a small change $\Delta b$ in the transformed space corresponds to a large change $\Delta x$ in the solution space.
\begin{figure}[S1]
\centering
\includegraphics[width=0.8\linewidth]{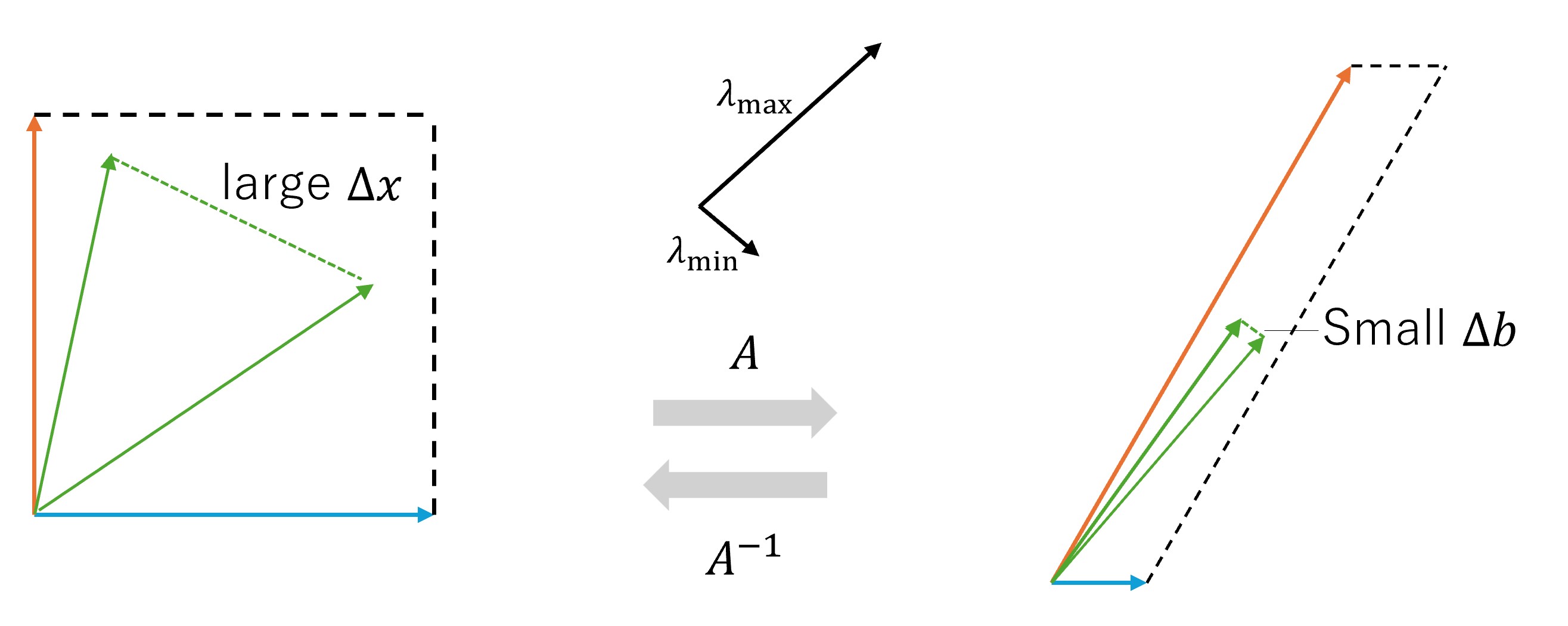}
\caption{Illustration of how the matrix $A$ transforms vectors.}
\label{fig:condition_number}
\end{figure}

For example, consider the symmetric positive definite matrix
\[
A = \begin{bmatrix} 2 & 1.9 \\ 1.9 & 2 \end{bmatrix}
\]
This matrix transforms vectors by stretching them $\lambda_{\max}$ = 3.9 times in one direction, while compressing them to $\lambda_{\min}$ = 0.1 times in another direction. The ratio of these stretching factors defines the condition number ${\rm cond}(A) = \lambda_{\max}/\lambda_{\min} = 39$. 
This large condition number implies that when solving $Ax = b$, a small relative error of 1$\%$ in $b$ can result in a much larger relative error of up to 39$\%$ in $x$.

This concept is generalized to higher dimensions, including the case of least-squares problems where $A$ is not a square matrix. In such cases, instead of eigenvalues, we use singular values of $A$ (which are the square roots of the eigenvalues of $A^TA$) as key indicators.
The following relationships help us understand how relative error propagation is governed by ${\rm cond}(A)$ in more general situations:
\begin{align*}
\|b\| &= \|Ax\|~(\leq \|A\|\|x\|) \quad \Rightarrow \quad\|x\| \geq \|A\|^{-1}\|b\|.
\end{align*}
When considering a perturbation $\Delta b$ in $b$ that leads to a change in the solution:
\begin{align*}
A\Delta x = \Delta b  \quad \Rightarrow \quad \|\Delta x\| \leq \|A^{+}\|\|\Delta b\|.
\end{align*}
Combining these inequalities yields:
\begin{align*}
\frac{\|\Delta x\|}{\|x\|} \leq \underbrace{\|A\|\|A^{+}\|}_{=\,{\rm cond}(A)}\frac{\|\Delta b\|}{\|b\|}.
\end{align*}
Note that when $A$ is a square matrix with an ordinary inverse, $A^{+}$ coincides with $A^{-1}$.
This inequality shows that even when the relative error in $b$ is small, it can be amplified by up to a factor of the condition number in the solution $x$.
Although the geometric interpretation in Fig. \ref{fig:condition_number} is not immediately obvious for nonsymmetric or rectangular matrices, this derivation shows why the condition number still governs the amplification of relative errors in the solution. For detailed theoretical treatment, see \cite{GV}.

\bibliography{iucr.bib}   
\bibliographystyle{iucr}     
\end{document}